\documentclass[english,hyphens]{article}
\usepackage[T1]{fontenc}
\usepackage[utf8]{inputenc}
\usepackage{color}
\usepackage{array}
\usepackage{refstyle}
\usepackage{float}
\usepackage{calc}
\usepackage{units}
\usepackage{mathtools}
\usepackage{multirow}
\usepackage{algorithm2e}
\usepackage{amsmath}
\usepackage{amsthm}
\usepackage{amssymb}
\usepackage{graphicx}
\PassOptionsToPackage{normalem}{ulem}
\usepackage{ulem}
\usepackage{microtype}

\makeatletter

\AtBeginDocument{\providecommand\thmref[1]{\ref{thm:#1}}}
\AtBeginDocument{\providecommand\eqref[1]{\ref{eq:#1}}}
\AtBeginDocument{\providecommand\algref[1]{\ref{alg:#1}}}
\providecommand{\tabularnewline}{\\}
\RS@ifundefined{subsecref}
  {\newref{subsec}{name = \RSsectxt}}
  {}
\RS@ifundefined{thmref}
  {\def\RSthmtxt{theorem~}\newref{thm}{name = \RSthmtxt}}
  {}
\RS@ifundefined{lemref}
  {\def\RSlemtxt{lemma~}\newref{lem}{name = \RSlemtxt}}
  {}

\numberwithin{equation}{section}
\numberwithin{figure}{section}
\theoremstyle{plain}
\newtheorem{thm}{\protect\theoremname}
\theoremstyle{definition}
\newtheorem{defn}[thm]{\protect\definitionname}
\theoremstyle{plain}
\newtheorem{conjecture}[thm]{\protect\conjecturename}
\theoremstyle{definition}
\newtheorem{example}[thm]{\protect\examplename}
\theoremstyle{remark}
\newtheorem{claim}[thm]{\protect\claimname}
\theoremstyle{plain}
\newtheorem{cor}[thm]{\protect\corollaryname}

\usepackage{float}
\usepackage{hyperref}
\hypersetup{colorlinks = false, allcolors=., citebordercolor = [rgb]{0,0,0}, linkbordercolor = [rgb]{0,0,0}, urlbordercolor = [rgb]{0,0,0}}
\usepackage[T1]{fontenc}
\usepackage{lipsum}
\usepackage{blindtext,titlefoot}
\usepackage{ragged2e}

\usepackage{pgfplots}
\pgfplotsset{width=7cm,compat=1.10}

\newenvironment{skproof}{%
  \proof}{\endproof}

\makeatother

\usepackage{babel}
\providecommand{\claimname}{Claim}
\providecommand{\conjecturename}{Conjecture}
\providecommand{\corollaryname}{Corollary}
\providecommand{\definitionname}{Definition}
\providecommand{\examplename}{Example}
\providecommand{\theoremname}{Theorem}

\begin{document}
\title{\textbf{Zero-Knowledge Proof-of-Identity}: {\Large{}Sybil-Resistant,
Anonymous Authentication on Permissionless Blockchains and Incentive
Compatible, Strictly Dominant Cryptocurrencies}}
\author{David Cerezo Sánchez\textsuperscript{}\\
{\small{}david@calctopia.com}}
\maketitle
\begin{abstract}
Zero-Knowledge Proof-of-Identity from trusted public certificates
(e.g., national identity cards and/or ePassports; eSIM) is introduced
here to permissionless blockchains in order to remove the inefficiencies
of Sybil-resistant mechanisms such as Proof-of-Work (i.e., high energy
and environmental costs) and Proof-of-Stake (i.e., capital hoarding
and lower transaction volume). The proposed solution effectively limits
the number of mining nodes a single individual would be able to run
while keeping membership open to everyone, circumventing the impossibility
of full decentralization and the blockchain scalability trilemma when
instantiated on a blockchain with a consensus protocol based on the
cryptographic random selection of nodes. Resistance to collusion is
also considered.

Solving one of the most pressing problems in blockchains, a zk-PoI
cryptocurrency is proved to have the following advantageous properties:

- an incentive-compatible protocol for the issuing of cryptocurrency
rewards based on a unique Nash equilibrium

- strict domination of mining over all other PoW/PoS cryptocurrencies,
thus the zk-PoI cryptocurrency becoming the preferred choice by miners
is proved to be a Nash equilibrium and the Evolutionarily Stable Strategy

- PoW/PoS cryptocurrencies are condemned to pay the Price of Crypto-Anarchy,
redeemed by the optimal efficiency of zk-PoI as it implements the
social optimum

- the circulation of a zk-PoI cryptocurrency Pareto dominates other
PoW/PoS cryptocurrencies

- the network effects arising from the social networks inherent to
national identity cards and ePassports dominate PoW/PoS cryptocurrencies

- the lower costs of its infrastructure imply the existence of a unique
equilibrium where it dominates other forms of payment


\pagebreak{}
\end{abstract}

\section{Introduction}

Sybil-resistance for permissionless consensus comes at a big price
since it needs to waste computation using Proof-of-Work (PoW), in
addition to assuming that a majority of the participants must be honest.
In contrast, permissioned consensus is able to overcome these issues
assuming the existence of a Public-Key Infrastructure\cite{aaba,cpps,pbft}
otherwise it would be vulnerable to Sybil attacks\cite{the-sybil-attack}:
indeed, it has been recently proved\cite{cryptoeprint:2018:302} that
consensus without authentication is impossible without using Proof-of-Work.
Proof-of-Stake, the alternative to PoW, is economically inefficient
because participants must keep capital at stake which incentivise
coin hoarding and ultimately leads to lower transaction volume.

Another major challenge in permissionless blockchains is scalability,
both in number of participants and total transaction volume. Blockchains
based on Proof-of-Work are impossible to scale because they impose
a winner-take-all contest between rent-seeking miners who waste enormous
amounts of resources, and their proposed replacements based on Proof-of-Stake
don't exhibit the high decentralization desired for permissionless
blockchains.

The solution proposed in this paper prevents Sybil attacks without
resorting to Proof-of-Work and/or Proof-of-Stake on permissionless
blockchains while additionally guaranteeing anonymous identity verification:
towards this goal, zero-knowledge proofs of trusted PKI certificates
(i.e., national identity cards and/or ePassports) are used to limit
the number of mining nodes that a single individual could run; alternatively,
a more efficient solution based on mutual attestation is proposed
and demonstrated practical \ref{subsec:Performance-Evaluation}. Counterintuitively,
the blockchain would still be permissionless even though using government
IDs because the term ``permissionless'' literally means ``without
requiring permission'' (i.e., to access, to join, ...) and governments
would not be authorizing access to the blockchain; moreover, the goal
is to be open to all countries of the world \ref{subsec:Worldwide-Coverage-and},
thus its openness is indistinguishable from PoW/PoS blockchains (i.e.,
the union of all possible national blockchains equals a permissionless,
open and global blockchain). Later papers in the literature \cite{leshno2024viabilityopensourcefinancialrails}
agree with this definition of permissionless (see Definition 1 of
\cite{leshno2024viabilityopensourcefinancialrails}) and further highlight
the importance of the use of identities in blockchain protocols (see
Section 5 of \cite{leshno2024viabilityopensourcefinancialrails}).
Coincidently, the latest regulations \cite{chinaIdentityRegulation,chinaRules}
point to the obligation to verify and use real-world identities on
blockchains, and the banning of contaminant cryptocurrency mining\cite{chinaBanMining,NDRCbanlist}.

Blockchain research has focused on better consensus algorithms obviating
that incentives are a central aspect of permissionless blockchains
and that better incentive mechanisms would improve the adoption of
blockchains much more that scalability improvements. To bridge this
gap, new proofs are introduced to demonstrate that mining a new cryptocurrency
based on Zero-Knowledge Proof-of-Identity would strictly dominate
previous PoW/PoS cryptocurrencies, thus replacing them is proved to
be a Nash equilibrium; additionally, the circulation of the proposed
cryptocurrency would Pareto dominate other cryptocurrencies. Furthermore,
thanks to the network effects arising from the network of users of
trusted public certificates, the proposed cryptocurrency could become
dominant over previous cryptocurrencies and the lower costs of its
infrastructure imply the existence of a unique equilibrium where it
dominates other forms of payment.

\subsection{Contributions}

The main and novel contributions are:
\begin{itemize}
\item The use of anonymous credentials in permissionless blockchains in
order to prevent Sybil attacks \ref{sec:Authentication-Protocols}:
previous works\cite{DBLP:journals/corr/KulynychITD17,blockchainCA}
considered the use of PKI infrastructures in blockchains (i.e., permissioned
ledgers) but without transforming them into anonymous credentials
in order to obtain the equivalent of a permissionless blockchain.
Other works have considered anonymous credentials on blockchains\cite{cryptoeprint:2013:622,coconut,DBLP:conf/ccs/CamenischDD17,quisquis,indyAnonCreds},
but requiring the issuance of new credentials and not reusing previously
existing ones: verifying real-world identities and issuing their corresponding
digital certificates is the most expensive part of any real-world
deployment.

\begin{itemize}
\item The practical implementation and its perfomance evaluation \ref{subsec:Performance-Evaluation}
for national identity cards and ePassports.
\end{itemize}
\item Circumventing the impossibility of full decentralization \ref{subsec:Circumventing-the-Impossibility}
and the blockchain scalability trilemma.
\item A protocol for an incentive-compatible cryptocurrency \ref{alg:Incentive-Compatible-Protocol}:
previous blockchains mint cryptocurrencies tied to the process of
reaching a consensus on the order of the transactions, but the game-theoretic
properties of this mechanism is neither clear nor explicit.
\item A proof that mining the proposed cryptocurrency is a dominant strategy
over other PoW/PoS blockchains and a Nash equilibrium over previous
cryptocurrencies \ref{subsec:Strictly-Dominant-Cryptocurrenci}, in
addition to an Evolutionary Stable Strategy \ref{subsec:Evolutionary-Stable-Strategies}.
\item The insight that the optimal efficiency of zk-PoI resides in that
it's implementing the social optimum, unlike PoW/PoS cryptocurrencies
that have to pay the Price of (Crypto-)Anarchy \ref{subsec:Price-of-Crypto-Anarchy}.
\item A proof that the circulation of the proposed zk-PoI cryptocurrency
Pareto dominates other PoW/PoS cryptocurrencies \ref{subsec:Pareto-Dominance-on}.
\item A proof that the proposed cryptocurrency could become dominant over
previous ones due to stronger network effects and the lack of acceptance
of previous cryptocurrencies as a medium of payment \ref{subsec:On-Network-Effects}.
\item Finally, the lower costs of its infrastructure imply the existence
of a unique equilibrium where it dominates other forms of payment
\ref{subsec:Dominance-over-Cash}.
\end{itemize}

\section{Related Literature}

This section discusses how the present paper is significantly better
and more innovative than previous approaches in order to fulfill the
objective of providing a Sybil-resistant and permissionless blockchain
with anonymous transaction processing nodes (i.e., miners). Moreover,
it's considerably cheaper than other approaches\cite{coconut,pop}
that would require the re-identification and issuing of new identities
to the global population because the current proposal relies on the
previously issued credentials of electronic national identity cards
(3.5 billion issued at the time of publication) and electronic passports
(1 billion issued at the time of publication).

Proof of Space\cite{cryptoeprint:2013:796,cryptoeprint:2013:805}
reduces the energy costs of Proof-of-Work but it's not economically
efficient. Proof of Authority\cite{poa}(PoA) maintains a public list
of previously authorised nodes: the identities are not anonymised
and the blockchain is not open to everyone (i.e., the blockchain is
permissioned). Proof of Personhood\cite{pop}(PoP) can be understood
as an improvement over Proof of Authority in that identities are anonymised,
but the parties/gatherings used to anonymise and incorporate identities
into the blockchain don't scale to national/international populations
and could compromise Sybil resistance because it's trivial to get
multiple identities by using different disguises on different parties/gatherings
(i.e., they need to be validated simultaneously and without disguises):
however, the present paper produces Sybil-resistant, anonymised identities
on a global scale for a permissionless blockchain. Moreover, Proof
of Personhood\cite{pop} is endogenizing all the costly process of
credential verification and issuing: by contrast, Zero-Knowledge Proof
of Identity is exogenizing/outsourcing this costly process to governments,
thus making the entire blockchain system cheaper. More recently, Private
Proof-of-Stake protocols\cite{cryptoeprint:2018:1105,cryptoeprint:2018:1132}(PPoS)
achieve anonymity, but the economic inefficiencies of staking capital
still remain and the identities have no relation to the real world.

A conceptually close work (``\textit{Decentralized Multi-authority
Anonymous Authentication for Global Identities with Non-interactive
Proofs}'', \cite{cryptoeprint:2019:701}), concurrently developed,
doesn't reuse real-world certificates and therefore it would require
that governments re-issue the cryptographic credentials of their citizens:
therefore, it doesn't consider neither Sybil-resistance nor blockchain
integration. Pseudo-anonymous signatures\cite{pseudoanonymousSignatures}
for identity documents provide an interesting technical solution to
the problem of anonymous authentication using identity documents.
However, the proposed schemes present a number of shortcomings that
discourage their use in the present setting: some schemes are closely
tied to particular countries (i.e., the German Identity Card\cite{cryptoeprint:2012:558,bsitr03110,cryptoeprint:2018:1148}),
thus non-general purpose enough to include any country in the world,
or flexible to adapt to future changes; they require interaction with
an issuer during card initialization; they feature protocols for deanonymisation
and revocation, not desired in the setting considered in this paper;
the initial German scheme\cite{bsitr03110} could easily be subverted\cite{10.1007/978-3-319-40367-0_31}
because the formalization of pseudo-anonymous signatures is still
incipient\cite{10.1007/978-3-319-49151-6_17}, and improvements are
being worked out\cite{cryptoeprint:2014:067,cryptoeprint:2016:070,cryptoeprint:2018:1148}.

Anonymous credentials, first envisioned by David Chaum\cite{10.1007/978-1-4757-0602-4_18},
and first fully realised by Camenisch and Lysyanskaya\cite{cryptoeprint:2001:019}
with follow-up work improving its security/performance\cite{cryptoeprint:2009:107,cryptoeprint:2005:060,cryptoeprint:2010:496,10.1007/978-3-540-28628-8_4,cryptoeprint:2012:298},
are a centrally important building block in e-cash. The use of anonymous
credentials to protect against Sybil attacks\cite{the-sybil-attack}
has already been proposed in previous works\cite{DBLP:conf/wistp/AnderssonKMP08,cryptoeprint:2007:384}
although with different cryptographic techniques and for different
goals. The main problem with anonymous credentials is that they require
a first identification step to an issuing party\cite{indyAnonCreds}
and that would compromise anonymity. This problem is shared with other
schemes for pseudoanymization: for example, Bitnym\cite{sybilPseudonym}
requires that a Trusted Third Party must check the real identity of
a user before allowing the creation of a bounded number of valid genesis
pseudonyms. Decentralized Anonymous Credentials\cite{cryptoeprint:2013:622}
was first to show how to decentralise the issuance of anonymous credentials
and integrate them within a blockchain (i.e., Bitcoin), but they do
not re-use previously existing credentials and they still rely on
Proof-of-Work for Sybil-resistance. Decentralized Blacklistable Anonymous
Credentials with Reputation\cite{cryptoeprint:2017:389} introduce
blacklistable reputation on blockchains, but users must also publish
their real-world identity (i.e., non-anonymous). QuisQuis\cite{quisquis}
introduces the novel primitive of updatable public keys in order to
provide anonymous transactions in cryptocurrencies, but it doesn't
consider their Sybil-resistance. DarkID\cite{darkID} is a practical
implementation of an anonymous decentralised identification system,
but requires non-anonymous pre-authentication and doesn't consider
Sybil-resistance. A previous work\cite{whoami} on secure identity
registration on distributed ledgers achieved anonymity from a credential
issuer, but the pre-authentication is non-anonymous, it doesn't consider
Sybil-resistance and it doesn't re-use real-world cryptographic credentials.
Recently, anonymous credentials on standard smart cards have been
proved practical\cite{cryptoeprint:2019:460}, but in a different
setting where the credential issuer and the verifier are the same
entity.

Previous works have also considered anonymous PKIs: for example, generating
pseudonyms\cite{anonymousPKI} using a Certificate Authority and a
separate Private Certificate Authority; however, this architecture
is not coherent for a permissionless blockchain because both certificate
authorities would be open to everyone and that would allow the easy
linking of anonymous identities. Another recent proposal for a decentralised
PKI based on a blockchain\cite{cryptoeprint:2018:853} does not provide
anonymity, although it improves the work on cryptographic accumulators
on blockchains started by Certcoin\cite{certcoin,cryptoeprint:2015:718};
another proposals introduce privacy-aware PKIs on blockchains\cite{pbpki,cryptoeprint:2019:527},
but they are not Sybil-resistant and do not re-use certificates from
other CAs. Previously, BitNym\cite{sybilPseudonym} introduced Sybil-resistant
pseudonyms to Bitcoin, but a Trusted-Third Party must check the real
identities of users before allowing the creation of a bounded number
of valid genesis pseudonyms. ChainAnchor\cite{chainAnchor} wasn't
Sybil-resistant and used Direct Anonymous Attestation just for anonymous
authentication, but not for mutual authentication: it worked on the
permissioned model, explicitly not permissionless, and the GroupOwner
initially knew the true identity of members; moreover, the Permissions
Issuer is supposed not to collude with the Verifier, although it has
reading access to the identity database. ClaimChain\cite{DBLP:journals/corr/KulynychITD17}
improves the decentralised distribution of public keys in a privacy-preserving
way with non-equivocation properties, but it doesn't consider their
Sybil-resistance because it's more focused on e-mail communications.
Blind Certificate Authorities\cite{WAPRS18,cryptoeprint:2016:925}
can simultaneously validate an identity and prove a certificate binding
a public key to it, without ever learning the identity, which sounds
perfect for the required scenario except that it requires 3 parties
and it's impossible to achieve in the 2-party setting; moreover, it
doesn't consider Sybil-resistance.

Other approaches to anonymous identity include: Lightweight Anonymous
Subscription with Efficient Revocation\cite{cryptoeprint:2018:290},
although it doesn't consider the real-world identity of users because
it's focused on the host and its Trusted Platform Module; One Time
Anonymous Certificates\cite{onetimeanonymouscertificate} extends
the X.509 standard to support anonymity through group signatures and
anonymous credentials, although it doesn't consider Sybil-resistance
and their group signatures require that users hold two group secret
keys, a requisite that is not allowed in the current scenario because
the user is not trusted to store them on the national identity card
(for the very same reasons, Linkable Ring Signatures\cite{cryptoeprint:2004:027}
and Linkable Message Tagging\cite{cryptoeprint:2014:014} are not
allowed as cryptographic tools whilst group signatures and Deniable
Anonymous Group Authentication\cite{daga} would require a non-allowed
setup phase). Opaak\cite{opaak} provides anonymous identities with
Sybil-resistance based on the scarcity of mobile phone numbers: however,
users must register by receiving an SMS message (i.e., the Anonymous
Identity Provider knows the real identity of participants). Oblivious
PRFs\cite{cryptoeprint:2018:733} are not useful in the permissionless
blockchains because the secret key of the OPRF would be known by everyone,
and the forward secrecy of the scheme that would provide security
even if the secret key is known would not be of any use because the
object identifiers ObjID would be easily predictable (i.e., derived
from national identifiers). SPARTA\cite{sparta} provides pseudonyms
through a distributed third-party-based infrastructure; however, it
requires non-anonymous pre-registration. UnlimitID\cite{unlimitid}
provides anonymity to OAuth/OpenID protocols, although users must
create keypairs and keep state between and within sessions, a requisite
that is not allowed in the current scenario. Another proposal for
anonymous pseudonyms with one Trusted-Third Party\cite{Yilek_traceableanonymous}
requires a division of roles between the TTP and the server that is
not coherent in a permissionless blockchain. With Self-Certified Sybil-free
Pseudonyms\cite{MKAP08,DBLP:conf/wistp/AnderssonKMP08}, the user
must keep state (i.e., dispenser D) generated by the issuer during
enrollment and the Sybil-free identification is based on unique featurs
of the devices, not on the user identity. Another anonymous authentication
using smart cards\cite{smartcardAnonAuth} is only anonymous from
an eavesdropping adversary, not from the authentication server itself.
TATA provides a novel way to achieve Sybil-resistant anonymous authentication:
members of an induction group must interact and keep a list of who
has already been given a pseudonym; therefore, a list of participants
could be collected, but they can't be linked to their real-world identities;
it's not clear how to bootstrap the initial set of trusted users to
get them to blindly sign each other's certificates.

Self-sovereign identity solutions usually rely on identities from
social networks, but their Sybil-resistance is very questionable because
almost half of their accounts could be fake\cite{facebookNYT}: in
spite of this, SybilQuorum\cite{sybilQuorum,sybilQuorumArxiv} proposes
the use of social network analysis techniques to improve their Sybil-resistance;
other research projects consider privacy-preserving cryptographic
credentials from federated online identities\cite{2014arXiv1406.4053M}.

Regarding the game-theoretic aspects, most papers focus on attacking
only one cryptocurrency (e.g., selfish mining\cite{DBLP:journals/corr/EyalS13},
miner's dilemma\cite{DBLP:journals/corr/Eyal14}, fork after withholding\cite{DBLP:journals/corr/abs-1708-09790}).
For a recent survey of these topics, see\cite{sokGameTheoryCryptocurrencies}.
Exceptionally, ``\textit{Game of Coins}''\cite{DBLP:journals/corr/abs-1805-08979}
considers the competition between multiple cryptocurrencies: a manipulative
miner alters coin rewards in order to move miners to other cryptocurrencies
of his own interest (with a fixed cost and a finite number of steps).
However, in this paper, it's the cryptocurrency issuer who changes
the rewards in order to attract miners from other cryptocurrencies
by producing the most efficient cryptocurrency to mine.

\textbf{}
\begin{table}[H]
\begin{centering}
\begin{tabular}{|c|c|c|c|c|c|c|c|}
\hline 
 & PoW & PoSpace & PoS & PPoS & PoA & PoP & \textbf{zk-PoI}\tabularnewline
\hline 
\hline 
(Pseudo)-Anonymity & \checkmark & \checkmark & $\times$ & \checkmark & $\times$ & \checkmark & \textbf{\checkmark}\tabularnewline
\hline 
Energy-Efficient & $\times$ & \checkmark & \checkmark & \checkmark & \checkmark & \checkmark & \textbf{\checkmark}\tabularnewline
\hline 
Economically Efficient & $\times$ & $\times$ & $\times$ & $\times$ & \checkmark & $\times$ & \textbf{\checkmark}\tabularnewline
\hline 
Permissionless & \checkmark & \checkmark & \checkmark & \checkmark & $\times$ & \checkmark ({*}) & \textbf{\checkmark}\tabularnewline
\hline 
\end{tabular}
\par\end{centering}
\textbf{\caption{Comparison of different Sybil-resistant mechanisms. \protect \\
({*}) Only if open to everyone, with no selective pre-invitation and
no right to exclude. \label{tab:Comparison-Sybil}}
}
\end{table}

\subsection{Proof-of-Personhood Considered Harmful (and Illegal)}

To be considered lawful in the real world, Proof-of-Personhood (PoP,
\cite{pop}) requires the concurrence of multiple unrestricted freedoms:
assembly, association, and wearing of masks. However, in most countries
these freedoms are limited:
\begin{itemize}
\item freedom of assembly\cite{freedomAssembly} and association\cite{freedomAssociation}:
most countries usually require previous notification and permission
from the governing authorities, that may reject for multiple grounds
including but not limited a breach of public order. Thus, PoP cannot
be considered permissionless in these countries.
\item it's forbidden to wear a mask in most countries\cite{antiMaskLaw},
as required for the anonimity of PoP (``All parties are recorded
for transparency, but attendees are free to hide their identities
by dressing as they wish, including hiding their faces for anonymity.'',
\cite{pop}). Thus, PoP won't be anonymous in countries that outlaw
the covering of faces.
\item promoters and organizers of PoP parties may themselves be committing
a crime, due to incitement, conspiracy and complicity.
\end{itemize}
The solution proposed in this paper it's the only possible lawful
one according to current regulations that require the use of national
IDs to register on blockchains (AMLD5\cite{AMLD5}, FATF\cite{fatfGuidance},
Cyberspace Administration of China\cite{chinaRules,chinaIdentityRegulation}).

\section{Building Blocks}

\subsection{Consensus based on the Cryptographic Random Selection of Transaction
Processing Nodes\label{subsec:Modern-Consensus-based-on}}

The new family of consensus algorithms based on the cryptographically
random selection of transaction processing nodes\cite{cryptoeprint:2016:918,cryptoeprint:2016:919,cryptoeprint:2017:406,cryptoeprint:2017:454,dfinityConsensus}
is characterised by:

\begin{flushleft}
\begin{tabular}{|c|>{\centering}p{5cm}|c|}
\hline 
Consensus algorithm & Random selection method & Sybil resistance\tabularnewline
\hline 
\hline 
OmniLedger & PVSS + collective BLS/BDN signatures \cite{cryptoeprint:2016:1067,cryptoeprint:2017:406,cryptoeprint:2018:483,cryptoeprint:2019:676} & PoW/PoS\tabularnewline
\hline 
RapidChain & Performance improvements over OmniLedger\cite{cryptoeprint:2018:460} & PoS\tabularnewline
\hline 
Algorand & Cryptographic sortition by a unique digital signature & PoS\tabularnewline
\hline 
Dfinity & BLS threshold signature scheme\cite{Boneh01shortsignatures}  & PoS\tabularnewline
\hline 
Snow White & Extract public keys based on the amount of currency owned & PoS\tabularnewline
\hline 
\end{tabular}
\par\end{flushleft}
\begin{itemize}
\item Transaction processing workers/nodes are randomly selected from a
larger group: in the case of Dfinity\cite{dfinityConsensus}, an unbiasable,
unpredictable verifiable random function (VRF) based on the BLS threshold
signature scheme\cite{Boneh01shortsignatures} with the properties
of uniqueness and non-interactivity; in the case of OmniLedger\cite{cryptoeprint:2017:406},
the original proposal used a collective Schnorr threshold signature
scheme\cite{DBLP:journals/corr/SytaTVWF15,cryptoeprint:2016:1067,cryptoeprint:2017:406},
although it has been updated to collective BLS/BDN signatures\cite{cryptoeprint:2018:483}
and now it uses MOTOR\cite{cryptoeprint:2019:676} instead of ByzCoin\cite{DBLP:journals/corr/Kokoris-KogiasJ16}
with improvements for open and public settings; in the case of Algorand,
secure cryptographic sortition is generated using an elliptic curve-based
verifiable random function (ECVRF-ED25519-SHA512-Elligator2\cite{algorandVRF});
in the case of Snow White, cryptographic committee reconfiguration
is done by extracting public keys from the blockchain based on the
amount of currency owned. For a detailed comparison of random beacon
protocols, see \cite{cryptoeprint:2018:319}.
\item Regular time intervals (also named epochs or rounds) on which randomly
selected workers/nodes process the transactions.
\item Faster transaction confirmation and finality.
\item High scalability.
\item Decoupling Sybil-resistance from the consensus mechanism (PoW/PoS
is about membership, not consensus).
\item PoW/PoS to protect against Sybil attacks: however, the present paper
proposes the use of Zero-Knowledge Proof-of-Identity (i.e., more economically\cite{stakedPoS}
and environmentally efficient\cite{natureEnergyCarbonCosts,bitcoinCarbonFootprint}).
\end{itemize}

\subsection{X.509 Public Key Infrastructure\label{subsec:X.509-Public-Key}}

X.509 is an ITU-T standard\cite{x509} defining the format of public
key certificates, itself based on the ASN.1 standard\cite{asn1}:
these certificates underpin most implementations of public key cryptography,
including SSL/TLS and smartcards. An X.509v3 certificate has the following
structure:

\begin{flushleft}
\noindent\fbox{\begin{minipage}[t]{1\columnwidth - 2\fboxsep - 2\fboxrule}%
\begin{itemize}
\item Certificate

\begin{itemize}
\item Version Number
\item Serial Number
\item Signature Algorithm ID
\item Issuer Name
\item Validity Period:

\begin{itemize}
\item Not Before
\item Not After
\end{itemize}
\item Subject Name
\item Subject Public Key Info

\begin{itemize}
\item Public Key Algorithm
\item Subject Public Key
\end{itemize}
\item Issuer Unique Identifier (optional)
\item Subject Unique Identifier (optional)
\item Extensions (optional):

\begin{itemize}
\item Key Usage (optional)
\item Authority Information Access (optional)
\item Certificate Policies (optional)
\item Basic Constraints (optional)
\item CRL Distribution Points (optional)
\item Subject Alternative Name (optional)
\item Extended Key Usage (optional)
\item Subject Key Identifier (optional)
\item Authority Key Identifier (optional)
\end{itemize}
\end{itemize}
\item Certificate Signature Algorithm
\item Certificate Signature
\end{itemize}
\end{minipage}}
\par\end{flushleft}

Certificates are signed creating a certificate chain: the root certificate
of an organization is a self-signed certificate that signs intermediate
certificates that themselves are used to sign end-entities certificates.
To obtain a signed certificate, the entity creates a key pair and
signs a Certificate Signing Request (CSR) with the private key: the
CSR contains the applicant's public key that is used to verify the
signature of the CSR and a unique Distinguished Name within the organization.
Then, one of the intermediate certificate authorities issues a certificate
binding a public key to the requested Distinguished Name and that
also contains information identifying the certificate authority that
vouches for this binding.

The certificate validation chain algorithm checks the validity of
an end-entity certificate following the next steps:
\begin{enumerate}
\item The certificates are correct regarding the ASN.1 grammar of X.509
certificates.
\item The certificates are within their validity periods (i.e., non-expired).
\item If access to a Certificate Revocation List is granted, the algorithm
checks that none of the certificates is included (i.e., the certificate
has not been revoked).
\item The certificate chain is traversed checking that:

\begin{enumerate}
\item The issuer matches the subject of the next certificate in the chain.
\item The signature is valid with the public key of the next certificate
in the chain.
\end{enumerate}
\item The last certificate is a valid self-signed certificated trusted by
the end-entity checker.
\end{enumerate}
Additionally, the algorithm could also check complex application policies
(i.e., the certificate can be used for web server authentication and/or
web client authentication).

\subsection{Electronic Passports\label{subsec:Electronic-Passport}}

\textbf{}
\begin{table}[H]
\begin{centering}
\begin{tabular}{|c|c|}
\hline 
\textbf{Data Group} & \textbf{Data Elements}\tabularnewline
\hline 
\hline 
\multirow{14}{*}{Data Group 1} & Document Types\tabularnewline
\cline{2-2} 
 & Issuing State or Organizaton\tabularnewline
\cline{2-2} 
 & Name (of Holder)\tabularnewline
\cline{2-2} 
 & Document Number\tabularnewline
\cline{2-2} 
 & Check Digit - Doc Number\tabularnewline
\cline{2-2} 
 & Nationality\tabularnewline
\cline{2-2} 
 & Date of Birth\tabularnewline
\cline{2-2} 
 & Check Digit - DOB\tabularnewline
\cline{2-2} 
 & Sex\tabularnewline
\cline{2-2} 
 & Date of Expiry or Valid Until Date\tabularnewline
\cline{2-2} 
 & Check Digit DOE/VUD\tabularnewline
\cline{2-2} 
 & Optional Data\tabularnewline
\cline{2-2} 
 & Check Digit - Optional Data Field\tabularnewline
\cline{2-2} 
 & Composite Check Digit\tabularnewline
\hline 
Data Group 11 & Personal Number\tabularnewline
\hline 
Data Group 15 & User's Public Key\tabularnewline
\hline 
\end{tabular}
\par\end{centering}
\textbf{\caption{Data Groups from Electronic Passports\label{tab:Data-Group-1}}
}
\end{table}

Modern electronic passports feature NFC chips\cite{icaoDoc9303part11}
that contain all their printed information in digital form, using
a proprietary format set by International Civil Aviation Organization\cite{icaoDoc9303part10}
and not X.509 certificates \ref{subsec:X.509-Public-Key} like the
ones used in national identity cards: the relevant fields are contained
within its Data Group 1 \ref{tab:Data-Group-1} (i.e., the same information
available within the Machine Readable Zone), and the Document Security
Object contains a hash of all the Data Groups signed by a Document
Signing Certificate issued every three months (also stored on the
passports), itself signed by a Country Signing Certificate Authority
(all the certificates are available online\cite{icaoPKD}). Additionally,
the data within the NFC chips are cryptographically protected and
it's necessary to derive the cryptographic keys by combining the passport
number, date of birth and expiry date (i.e., BAC authentication).

Finally, note that the electronic identity cards of some countries
can also work as ePassports (e.g., Spanish Identity Card -Documento
Nacional de Identidad-).

\subsection{Verifiable Computation\label{subsec:Verifiable-Computation}}

A public verifiable computation scheme allows a computationally limited
client to outsource to a worker the evaluation of a function $F\left(u,w\right)$
on inputs $u$ and $w$: other alternative uses of these schemes allow
a verifier $V$ to efficiently check computations performed by an
untrusted prover $P$. More formally, the following three algorithms
are needed:
\begin{defn}
(Public Verifiable Computation). A public verifiable computation scheme
$VC$ consists of three polynomial-time algorithms $\left(\mbox{Keygen},\mbox{ Compute, Verify}\right)$
defined as follows:
\end{defn}

\begin{itemize}
\item $\left(EK_{F},VK_{F}\right)\leftarrow\mbox{Keygen}\left(F,1^{\lambda}\right)$:
the key generation algorithm takes the function $F$ to be computed
and security parameter $\lambda$; it outputs a public evaluation
key $EK_{F}$ and a public verification key $VK_{F}$.
\item $\left(y,\pi_{y}\right)\leftarrow\mbox{Compute}\left(EK_{F},u,w\right)$:
the prover runs the deterministic worker algorithm taking the public
evaluation key $EK_{F}$, an input $u$ supplied by the verifier and
an input $w$ supplied by the prover. It outputs $y\leftarrow F\left(u,w\right)$
and a proof $\pi_{y}$ of $y$'s correctness (as well as of prover's
knowledge of $w$).
\item $\left\{ 0,1\right\} \leftarrow\mbox{Verify}\left(VK_{F},u,w,y,\pi_{y}\right)$:
the deterministic verification algorithm outputs $1$ if $F\left(u,w\right)=y$,
and $0$ otherwise.
\end{itemize}
A public verification computation scheme $VC$ must comply with the
following properties of correctness, security, and efficiency:
\begin{itemize}
\item Correctness: for any function $F$ and any inputs \uline{\mbox{$u,w$}}
to $F$, if we run $\left(EK_{F},VK_{F}\right)\leftarrow\mbox{Keygen}\left(F,1^{\lambda}\right)$
and $\left(y,\pi_{y}\right)\leftarrow\mbox{Compute}\left(EK_{F},u,w\right)$
then we always get $\mbox{Verify}\left(VK_{F},u,w,y,\pi_{y}\right)=1$.
\item Efficiency: $\mbox{Keygen}$ is a one-time setup operation amortised
over many calculations and $Verify$ is computationally cheaper than
evaluating $F$.
\item Security: for any function $F$ and any probabilistic polynomial-time
adversary $A$, we require that
\[
\mbox{Pr}\left[\left(\hat{u},\hat{w},\hat{y},\hat{\pi_{y}}\right)\leftarrow A\left(EK_{F},VK_{F}\right):F\left(\hat{u},\hat{w}\right)\neq\hat{y}\right]\leq\mbox{negl}\left(\lambda\right)
\]
and
\[
1=\mbox{Verify}\left(VK_{F},\hat{u},\hat{w},\hat{y},\hat{\pi_{y}}\right)\leq\mbox{negl}\left(\lambda\right)
\]
where $\mbox{negl}\left(\lambda\right)$ denotes a negligible function
of inputs $\lambda$.
\end{itemize}
Additionally, we require the public verification computation scheme
$VC$ to be succinct and zero-knowledge:
\begin{itemize}
\item Succinctness: the generated proofs $\pi_{y}$ are of constant size,
that is, irrespective of the size of the function $F$ and inputs
$u$ and $w$.
\item Zero-knowledge: the verifier learns nothing about the prover's input
$w$ beyond the output of the computation.
\end{itemize}
Practical implementations are Pinocchio\cite{cryptoeprint:2013:279}
and Geppeto\cite{cryptoeprint:2014:976}, or Buffet\cite{cryptoeprint:2014:674}
and Pequin\cite{pequin}(a simplified version of Pepper\cite{Setty12makingargument}).

\subsubsection{Verifiable Validation of X.509 Certificates as Anonymous Credentials\label{subsec:Verifiable-Validation-of-X509Certificates}}

The algorithm for certificate chain validation chain in section \ref{subsec:X.509-Public-Key}
can be implemented with the public verifiable computation scheme of
section \ref{subsec:Verifiable-Computation} using zk-SNARKS to obtain
a verifiable computation protocol so that a certificate holder is
able to prove that he holds a valid X.509 certificate chain with a
unique Distinguished Name, without actually sending the public key
to the verifier and selectively disclosing the contents of the certificate:
in other words, we re-use existing certificate chains and PKI infrastructure
without requiring any modifications, turning X.509 certificates into
anonymous credentials. A previous work already demonstrated the technical
and practical viability of this approach\cite{cinderella-turning-shabby-x-509-certificates-into-elegant-anonymous-credentials-with-the-magic-of-verifiable-computation}:
the only handicap was that the proof generation could take a long
time (e.g., more than 10 minutes) and large keys (e.g., 1 Gbyte)..

Recent research advances have improved\cite{cryptoeprint:2017:602}
the initial setup of the zk-SNARK protocol used to generate the Common-Reference
String (CRS) with an MPC protocol, such that it's secure even if all
participants are malicious (except one). And faster proving times
could be obtained by efficiently composing the non-interactive proving
of algebraic and arithmetic statements\cite{cryptoeprint:2018:557}
since QAP-based zk-SNARKs are only efficient for arithmetic representations
and not algebraic statements, but at the cost of increasing the proof
size.

In this paper, a practical implementation was completed to check a
certificate chain with an additional validation policy and written
as C code for Pequin\cite{pequin}, then compiled into a public evaluation
and verification keys: unfortunately, it isn't scalable to millions
of users and/or the large circuits/constraints required to cover all
the typologies of national identity cards/ePassports, thus an implementation
based on TEE and mutual attestation is the preferred implementation
\ref{subsec:Detailed-Authentication-Protocol-Remote-Attestation}.
The only zero-knowledge proof system that could be scalable enough\cite{cryptoeprint:2018:691}
works on a computer cluster, thus it doesn't fit the setting of a
single user authenticating on his own device, and a libsnark backend
can't handle more than 4 million gates requiring more than an hour
of computation.

\subsection{Cryptographic Accumulators}

Firstly devised by Benaloh and de Mare\cite{10.1007/3-540-48285-7_24},
a cryptographic accumulator \cite{cryptoeprint:2015:087} is a compact
binding set of elements supporting proofs of membership and more space-efficient
than storing all of the elements of the set; given an accumulator,
an element, and a membership witness, the element's presence in the
accumulated set can be verified. Generally speaking, an accumulator
consists of four polynomial-time algorithms:
\begin{itemize}
\item $Generate\left(1^{k}\right)$: given the security parameter $k$,
it instantiates the initial value of the empty accumulator.
\item $Add\left(a,y\right)\rightarrow\left(a',w\right)$: adds the element
$y$ to the current state of the accumulator $a$ producing the updated
accumulator value $a'$ and the membership witness $w$ for $y$.
\item $WitnessAdd\left(w,y\right)\rightarrow w'$: on the basis of the current
state of a witness $w$ and the newly added value $y$, it returns
an updated witness $w'$.
\item $Verify\left(a,y,w\right)\rightarrow\left\{ true,false\right\} $:
verifies the membership of $y$ using its witness $w$ on the current
state of accumulator $a$.
\end{itemize}
The following are interesting security properties of accumulators:
\begin{itemize}
\item Dynamic accumulators\cite{10.1007/3-540-45708-9_5}: accumulators
supporting the removal of elements from the accumulator by means of
a deletion algorithm $Removal()$ and a witness update algorithm $WitnessRemoval\left(\right)$.
\item Universality\cite{10.1007/978-3-540-72738-5_17}: accumulators supporting
non-membership proofs, $NonWitnessAdd\left(\right)$, $NonWitnessRemoval\left(\right)$
and $NonVerify\left(\right)$.
\item Strong accumulators\cite{Camacho2012}: deterministic and publicly
executable, meaning that it does not rely on a trusted accumulator
manager.
\item Public checkable accumulators, the correctness of every operation
can be publicly verified.
\end{itemize}
Recent constructions of cryptographic accumulators specifically tailored
for blockchains are: a dynamic, universal, strong and publicly checkable
accumulator \cite{certcoin}; an asynchronous accumulator\cite{cryptoeprint:2015:718}
with low frequency update and old-accumulator compatibility (i.e.,
up-to-date witnesses can be verified even against an outdated accumulator);
a constant-sized, fair, public-state, additive, universal accumulator\cite{cryptoeprint:2018:853},
and an accumulator optimised for batch and aggregation operations\cite{cryptoeprint:2018:1188}.

\subsection{Remote Attestation}

In the terminology of Intel SGX, remote attestation is used to prove
that an enclave has been established without alterations of any kind:
in other words, remote parties can verify that an application is running
inside an SGX enclave. Concretely, remote attestation is used to verify
three properties: the identity of the application, that it has not
been tampered with, and that it is running securely within an SGX
enclave. Remote attestation is carried out in several stages: requesting
a remote attestation from the challenger; performing a local attestation
of the enclave; converting said local attestation to a remote attestation;
returning the remote attestation to the challenger, and the challenger
verifying the remote attestation to the Intel Attestation Service. 

A detailed technical description is outside of the scope of this paper:
detailed descriptions can be found in the standard technical documentation\cite{cryptoeprint:2016:086,officialAttestation,sampleAttestation}.
Recent attacks\cite{vanbulck2018foreshadow} can be used to extract
the secret attestation keys used to verify the identity of an SGX
enclave, and microcode updates must be installed\cite{intelL1TF}
to prevent their exploitation: that is, it's essential to check that
parties to a remote attestation are using a safe and updated version.
However, our protocols are inherently resistant to deniability attacks\cite{cryptoeprint:2018:424}
because they are based on mutual attestation.

As it would be shown in the next section \ref{subsec:Detailed-Authentication-Protocol-Remote-Attestation},
remote attestation can be used as a more efficient substitute of verifiable
computation.

\section{Authentication Protocols\label{sec:Authentication-Protocols}}

In this section, we describe authentication protocols for Sybil-resistant,
anonymous authentication using Zero-Knowledge protocols \ref{subsec:Detailed-Authentication-Protocol}
and remote attestation \ref{subsec:Detailed-Authentication-Protocol-Remote-Attestation}.

\subsection{Authentication Protocols using Zero-Knowledge\label{subsec:Detailed-Authentication-Protocol}}

\textcolor{black}{The use of zero-knowledge protocols guarantee the
public-verifibility of the correctness of the Sybil-resistant, anonymous
identities committed to the permissionless blockchain.}

\subsubsection{Security Goals}

The following security goals must be met for the system to be considered
secure:
\begin{enumerate}
\item The registered miner's key to the blockchain \textit{opens, but no
one can shut; he can shut, but no one can open} (\textit{Isaiah 22:22},
\cite{isaiah2222}). For the security of the system to be considered
equivalent to the currently available permissionless blockchains,
anyone holding a valid public certificate should be able to register
a pseudo-anonymous identity on the blockchain but no one should be
able to remove it (i.e., uncensorable free entry is guaranteed).
\item Protection against malicious issuers: some certification authorities
may turn against some citizens and try to cancel access to the permissionless
blockchain or stole their funds.

\begin{enumerate}
\item Mandatory passphrase. An issuer may counterfeit a certificate with
the same unique identifiers, thus possessing a valid certificate isn't
secure enough and a passphrase is deemed mandatory.
\item Non-bruteforceable. Operations must be computationally costly on the
client side to prevent brute-forcing.
\item No OCSP checking. Prevention against malicious blacklisting.
\end{enumerate}
\item Privacy: miner's real identity can't be learned by anyone.
\item Unique pseudonyms: from each identity card/ePassport, only one unique
identifier can be generated.
\item Publicly verifiable: anyone should be able to verify the validity
of the miner's public key and its pseudonym.
\end{enumerate}

\subsubsection{Zero-Knowledge Protocols (X.509)\label{subsec:Zero-Knowledge-Protocols-(X.509)}}

\textbf{Anonymous miner registration of a new public key on a permissionless
blockchain. }This protocol generates a unique pseudonym for each miner,
and attaches a verifiable proof that its new public key to be stored
on-chain is signed with a valid public certificate included on a recognised
certification authorities list, and that the new public key is linked
to the blockchain-specific pseudonym that is in turn uniquely linked
to the citizen's public key certificate.

Miners holding a public key certificate must execute the following
steps:
\begin{enumerate}
\item Create a deterministic public/secret key pair based on a secret passphrase
(no need for verifiable computation):
\[
pk,sk=\mbox{Det\_KeyPairGen}\left(KDF\left(passphrase,hash(publicCert)\right)\right)
\]
The generation algorithm must be determistic because the smartcard
may be unable to store them and/or the miner may loose them (i.e.,
as in deterministic wallets). KDF is a password-based key derivation
function (e.g., PBKDF2).
\item Obtain a signature of the previously generated public key $pk$ with
the miner's public key certificate (no need for verifiable computation,
this operation could be executed on a smartcard):
\[
sign_{PK}=\mbox{PKCS\_Sign}\left(secretKey_{publicCert},pk\right)
\]
\item Check the validity of the certificate chain of the miner's public
key certificate as extracted from the smartcard:

\begin{enumerate}
\item Load the public key of the root certificate.
\item Hash and verify all intermediates, based on their certificate templates,
and the public key of their parent certificate starting from the root
certificate and following with the verified public key from the previous
intermediate certificate template.
\item Hash and verify the miner's public key certificate using the last
verified public key returned from the previous step.
\item Check the time validity of the miner's public key certificate.
\item Check that the miner's public key certificate is contained on a list
of trusted certification authorities.
\end{enumerate}
\item Obtain the unique identifier from the miner's public key:
\[
uniqueID=getID(publicCert)
\]
Note that the unique identifier is usually contained on Serial Number
of the certificate, or the Subject Alternative Name extension under
different OIDs, depending on the country.
\item Generate a deterministic pseudonym using the blockchain identifier:
\begin{eqnarray*}
signatureSecret & = & \mbox{PKCS\_Sign}\left(secretKey_{publicCert},\right.\\
 &  & \left."\mbox{PREFIXED\_COMMON\_STRING}"\right)
\end{eqnarray*}
\begin{eqnarray*}
pseudonym & = & Hash\left(signatureSecret||BlockchainIdentifier||uniqueID\right)\\
 &  & ||"\mbox{REG}"
\end{eqnarray*}
PKCS\_Sign is the deterministic PKCS\#1.5 signing algorithm executed
on a prefixed string to obtain a unique, non-predictable secret based
on the certificate's owner. The obtained signature is appended to
the blockchain identifer and the unique identifier, and then hashed
to derive a unique pseudonym. Finally, the string ``REG'' is appended
to differentiate this pseudonym from the one generated during a remove
protocol and prevent replay attacks for removal reusing the generated
zero-knowledge proof.
\item Verify the signature $sign_{PK}$ on the miner's public key certificate
$pk$: 
\[
\mbox{PKCS\_Verify}\left(publicCert,sign_{PK}\right)
\]
\item As the $signatureSecret$ is calculated offline by the smartcard,
it's also necessary to verify it using the miner's public key certificate
$publicCert$:
\[
\mbox{PKCS\_Verify}\left(publicCert,signatureSecret\right)
\]
\item Generate the zero-knowledge proof $\pi$ (e.g., zk-SNARK, zk-STARK
or zk-SNARG) of the miner's public key certificate $pk$, the generated
pseudonym and, signature $sign_{PK}$ such that all the previous conditions
3-7 hold.
\item Anonymously contact the permissionless blockchain: 

\begin{enumerate}
\item optionally, check the miner's real identity on a cryptographic accumulator:

\begin{enumerate}
\item establish a shared secret running a Diffie-Hellman key exchange between
the prospective miner and the permissionless blockchain
\item send attributes of the miner's real identity encrypted with the shared
secret
\item execute the non-membership proof $NonWitnessAdd\left(w,y\right)$
on the cryptographic accumulator
\end{enumerate}
\item register the generated pseudonym, the new public key $pk$, the signature
$sign_{PK}$ and $\pi$: note that they don't reveal the miner's real
identity ($publicCert$, $uniqueID$ and $signatureSecret$ are all
keep as a secret).
\end{enumerate}
\end{enumerate}
The registering node of the permissionless blockchain verifies $\pi$
before adding the new public key, the associated pseudonym, the signature
$sign_{PK}$ and the succinct proof $\pi$: note that the miner is
unable to register multiple pseudonyms, and he can only use one running
node that would be signing messages with the generated secret key
$sk$. Other nodes would be able to efficiently verify $\pi$ to confirm
that the public key $pk$ is a signed by someone from an allowed certificate
authority, and that the pseudonym is the miner's unique alias for
the blockchain.\\

\textbf{Taking offline registrations from a permissionless blockchain.
}This protocol takes offline a pseudonym and its associated public
key $pk$ and signature $sign_{PK}$ from a permissionless blockchain.
Miners must execute the following steps to take offline an identity
from a permissionless blockchain:
\begin{enumerate}
\item Generate a zero-knowledge proof $\pi$ (e.g., zk-SNARK, zk-STARK or
zk-SNARG) of the steps 3-7 of the previous protocol to prove secret
knowledge of $sk$ and that he's able to re-generate the pseudonym,
but this time appending the string ``OFF'' to the pseudonym.
\item Anonymously contact the permissionless blockchain to take offline
the generated pseudonym and all its associated data (including the
cryptographic accumulator), attaching $\pi$.
\end{enumerate}
The registering node of the permissionless blockchain verifies $\pi$
before taking offline the pseudonym without learning the real identity
of the miner (publicCert, uniqueID and signatureSecret remain secret).

\subsubsection{Zero-Knowledge Protocols (ePassports)}

Analogous to the zero-knowledge protocols for X.509 \ref{subsec:Zero-Knowledge-Protocols-(X.509)},
but now considering the specific details of ePassports \ref{subsec:Electronic-Passport},
which usually contain a unique keypair with the public key on Data
Group 15 and the private key hidden within the chip: the Active Authentication
protocol can be used to sign random challenges that can be verified
with the corresponding public key. Some ePassports don't feature Active
Authentication, nonetheless a modified version of the following protocols
could still be executed (see subsection \ref{subsec:zkAbsence-of-AA}).

\textbf{Anonymous miner registration of a new public key on a permissionless
blockchain. }This protocol generates a unique pseudonym for each miner,
and attaches a verifiable proof that its new public key to be stored
on-chain is signed with a valid public certificate included on the
list of Country Signing Certificate Authorities, and that the new
public key is linked to the blockchain-specific pseudonym that is
in turn uniquely linked to the public key certificate of the passport
holder.

Miners holding a public key certificate must execute the following
steps:
\begin{enumerate}
\item Create a deterministic public/secret key pair based on a secret passphrase
(no need for verifiable computation):
\[
pk,sk=\mbox{Det\_KeyPairGen}\left(KDF\left(passphrase,hash(publicCert)\right)\right)
\]
The \textit{publicCert} is taken from the Data Group 15. KDF is a
password-based key derivation function (e.g., PBKDF2).
\item Obtain a signature of the previously generated public key $pk$ with
the miner's public key certificate (no need for verifiable computation,
this operation is executed within the ePassport's chip using the Active
Authentication protocol):
\[
sign_{PK}=\mbox{Sign}\left(secretKey_{publicCert},pk\right)
\]
\item Check the validity of the Data Security Object of the miner's ePassport:

\begin{enumerate}
\item Load the public key of the Country Signing Certificate from a trusted
source \cite{icaoPKD} and the Document Signing Certificate from the
ePassport.
\item Hash all the Data Groups and check their equivalence to the Data Security
Object.
\item Verify the signature of the Data Security Object using the Document
Signing Certificate.
\item Verify the signature of the Document Signing Certificate using the
Country Signing Certificate.
\item Check the time validity of the certificates.
\end{enumerate}
\item Obtain the unique identifier of the ePassport:
\[
uniqueID=getID(DataGroups)
\]
Note that the unique identifier is usually contained on the Data Element
``Document Number'' of the Data Group 1: as it's legally valid for
the same person to own multiple passports with different Document
Numbers, some countries include a unique ``Personal Number'' on
the Data Group 11.
\item Generate a deterministic pseudonym using the blockchain identifier:
\begin{eqnarray*}
signatureSecret & = & \mbox{Sign}\left(secretKey_{publicCert},\right.\\
 &  & \left."\mbox{PREFIXED\_COMMON\_STRING}"\right)
\end{eqnarray*}
\begin{eqnarray*}
pseudonym & = & Hash\left(signatureSecret||BlockchainIdentifier||uniqueID\right)\\
 &  & ||"\mbox{REG}"
\end{eqnarray*}
Sign is the Active Authentication protocol executed within the ePassport's
chip, a deterministic signing algorithm executed on a prefixed string
to obtain a unique, non-predictable secret based on the certificate's
owner. The obtained signature is appended to the blockchain identifer
and the unique identifier, and then hashed to derive a unique pseudonym.
Finally, the string ``REG'' is appended to differentiate this pseudonym
from the one generated during a remove protocol and prevent replay
attacks for removal reusing the generated zero-knowledge proof (e.g.,
zk-SNARK, zk-STARK or zk-SNARG).
\item Verify the signature $sign_{PK}$ on the miner's public key certificate
$pk$: 
\[
\mbox{PKCS\_Verify}\left(publicCert,sign_{PK}\right)
\]
The \textit{publicCert} is taken from the Data Group 15.
\item As the $signatureSecret$ is calculated offline by the ePassport's
chip, it's also necessary to verify it using the miner's public key
certificate $publicCert$:
\[
\mbox{PKCS\_Verify}\left(publicCert,signatureSecret\right)
\]
The \textit{publicCert} is taken from the Data Group 15.
\item Generate the zero-knowledge proof $\pi$ (e.g., zk-SNARK, zk-STARK
or zk-SNARG) of the miner's public key certificate $pk$, the generated
pseudonym, and signature $sign_{PK}$ such that all the previous conditions
3-7 hold.
\item Anonymously contact the permissionless blockchain

\begin{enumerate}
\item optionally, check the miner's real identity on a cryptographic accumulator:

\begin{enumerate}
\item establish a shared secret running a Diffie-Hellman key exchange between
the prospective miner and the permissionless blockchain
\item send attributes of the miner's real identity encrypted with the shared
secret
\item execute the non-membership proof $NonWitnessAdd\left(w,y\right)$
on the cryptographic accumulator
\end{enumerate}
\item register the generated pseudonym, the new public key $pk$, the signature
$sign_{PK}$ and $\pi$: note that they don't reveal the miner's real
identity ($publicCert$, $uniqueID$ and $signatureSecret$ are all
keep as a secret).
\end{enumerate}
\end{enumerate}
The registering node of the permissionless blockchain verifies $\pi$
before adding the new public key, the associated pseudonym, the signature
$sign_{PK}$ and the succinct proof $\pi$: note that the miner is
unable to register multiple pseudonyms, and he can only use one running
node that would be signing messages with the generated secret key
$sk$. Other nodes would be able to efficiently verify $\pi$ to confirm
that the public key $pk$ is a signed by someone from an allowed certificate
authority and that the pseudonym is the miner's unique alias for the
blockchain.\\

\textbf{Taking offline registrations from a permissionless blockchain.
}This protocol takes offline a pseudonym and its associated public
key $pk$ and signature $sign_{PK}$ from a permissionless blockchain.
Miners must execute the following steps to take offline an identity
from a permissionless blockchain:
\begin{enumerate}
\item Generate a zero-knowledge proof $\pi$ (e.g., zk-SNARK, zk-STARK or
zk-SNARG) of the steps 3-7 of the previous protocol to prove secret
knowledge of $sk$ and that he's able to re-generate the pseudonym,
but this time appending the string ``OFF'' to the pseudonym.
\item Anonymously contact the permissionless blockchain to take offline
the generated pseudonym and all its associated data (including the
cryptographic accumulator), attaching $\pi$.
\end{enumerate}
The registering node of the permissionless blockchain verifies $\pi$
before taking offline the pseudonym without learning the real identity
of the miner (publicCert, uniqueID and signatureSecret remain secret).

\subsubsection{Mapping to goals}

The previous protocols achieve the security goals:
\begin{enumerate}
\item The registered miner's key to the blockchain \textit{opens, but no
one can shut; he can shut, but no one can open}. Only someone in possession
of a valid public certificate can create a unique miner identity on
the open blockchain and destroy it. Please note that the signing and
verification of steps 2, 5, 6 and 7 are only needed if it's required
to check that the miner is the real owner of the smartcard/ePassport.
\item Protection against malicious issuers: the passphrase is mandatory,
there's no OCSP checking and the protocol is non-bruteforceable because
it requires the generation of a proof $\pi$ for every passphrase
that is going to be tried (>60 secs per $\pi$).
\item Privacy: miner's real identity can't be learned by anyone because
publicCert and uniqueID are keep secret.
\item Unique pseudonyms: from each identity card/ePassport, only one unique
identifier can be generated because there's only one uniqueID per
citizen.
\item Publicly verifiable: using the proof $\pi$, anyone is able to validate
the miner's public key and its pseudonym.
\end{enumerate}
Additionally, cryptographic accumulators could be added to the protocols
in order to prevent multiple registrations whenever an expired certificate
is renovated.

\subsubsection{Absence of Active Authentication\label{subsec:zkAbsence-of-AA}}

Signing using the secret key of the Active Authentication protocol
provides an extra layer of security: it guarantess that the remote
party executing the protocol owns a physical copy of the ePassport
(i.e., it hasn't stolen a copy of the public certificates from others).
However, some ePassports don't feature Active Authentication, requiring
a simplified version of the previous protocols:
\begin{itemize}
\item Steps 2,6 and 7 are removed.
\item Step 5 doesn't calculate the signature.
\item The zero-knowledge$\pi$ is extended to Step 1, with a password-based
key derivation function using less steps.
\end{itemize}

\subsection{Detailed Authentication Protocols using Mutual Attestation\label{subsec:Detailed-Authentication-Protocol-Remote-Attestation}}

\textcolor{black}{The use of remote attestation protocols guarantee
the efficiency and scalability of the full authentication solution
(i.e., it can easily scale to billions of users). By design, the architecture
has detached the encrypted DB from the mining nodes to maintain the
implementation as blockchain-agnostic as possible: some mining nodes
may include an encrypted DB, but it's not necessary that all mining
nodes include it.}

\begin{figure}[H]
\includegraphics[scale=0.8]{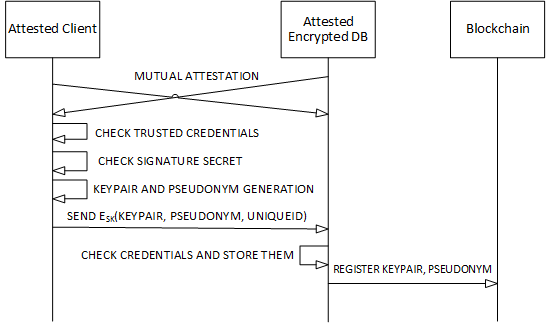}

\caption{\label{fig:Legend:-(1)-Attestation}Simplified overview of mutual
attestation.}
\end{figure}

\subsubsection{Security Goals}

The following security goals must be met for the system to be considered
secure:
\begin{enumerate}
\item The registered miner's key to the blockchain \textit{opens, but no
one can shut; he can shut, but no one can open} (\textit{Isaiah 22:22},
\cite{isaiah2222}). For the security of the system to be considered
equivalent to the currently available permissionless blockchains,
anyone holding a valid public certificate should be able to register
a pseudo-anonymous identity on the blockchain but no one should be
able to remove it (i.e., uncensorable free entry is guaranteed).
\item Protection against malicious issuers: some certification authorities
may turn against some citizens and try to cancel access to the permissionless
blockchain or stole their funds.

\begin{enumerate}
\item Mandatory passphrase. An issuer may counterfeit a certificate with
the same unique identifiers, thus possessing a valid certificate isn't
secure enough and a passphrase is deemed mandatory.
\item Non-bruteforceable. Operations must be computationally costly on the
client side to prevent brute-forcing.
\item No OCSP checking. Prevent against malicious blacklisting.
\end{enumerate}
\item Privacy: miner's real identity can't be learned by anyone.
\item Unique pseudonyms: from each identity card/ePassport, only one unique
identifier can be generated.
\end{enumerate}

\subsubsection{Mutual Attestation for X.509 Certificates\label{subsec:Mutual-Attestation-X509}}

\textbf{Anonymous miner registration of a new public key on a permissionless
blockchain.} This protocol generates a unique pseudonym for each miner,
with a new public key linked to the blockchain-specific pseudonym
that is in turn uniquely linked to the citizen's public key certificate:
the mutual attestation between the parties guarantees the correctness
of the execution of both parties.

The following are the steps to the protocol:
\begin{enumerate}
\item The client locally generates a signature secret using its secret key:
\begin{eqnarray*}
signatureSecret & = & \mbox{PKCS\_Sign}\left(secretKey_{publicCert},\right.\\
 &  & \left."\mbox{PREFIXED\_COMMON\_STRING}"\right)
\end{eqnarray*}
\item Mutual attestation between the authenticating client and the blockchain:
the attestation is anonymous thanks to the use of unlinkable signatures
(Enhanced Privacy ID -EPID-), and both parties obtain a temporary
secret key to encrypt their communications.
\item Client's attested code checks the validity of the certificate chain
of the miner's public key certificate as extracted from the smartcard:

\begin{enumerate}
\item Load the public key of the root certificate.
\item Hash and verify all intermediates, based on their certificate templates,
and the public key of their parent certificate starting from the root
certificate and following with the verified public key from the previous
intermediate certificate template.
\item Hash and verify the miner's public key certificate using the last
verified public key returned from the previous step.
\item Check the time validity of the miner's public key certificate.
\item Check that the miner's public key certificate is contained on a list
of trusted certification authorities.
\end{enumerate}
\item If the previous step concluded satisfactorily, then the client's attested
code verifies the $signatureSecret$ using the miner's public key
certificate $publicCert$ because the $signatureSecret$ is calculated
offline by the smartcard:
\[
\mbox{PKCS\_Verify}\left(publicCert,signatureSecret\right)
\]
\item If the previous step concluded satisfactorily, then the client's attested
code creates a deterministic public/secret key pair based on a secret
passphrase:
\[
pk,sk=\mbox{Det\_KeyPairGen}\left(KDF\left(passphrase,hash(publicCert)\right)\right)
\]
The generation algorithm must be deterministic because the smartcard
may be unable to store them and/or the miner may lose them (i.e.,
as in deterministic wallets). KDF is a password-based key derivation
function (e.g., PBKDF2).
\item The client's attested code generates a deterministic pseudonym using
the blockchain identifier:
\begin{eqnarray*}
pseudonym & = & Hash\left(signatureSecret||BlockchainIdentifier||uniqueID\right)\\
 &  & ||"\mbox{REG}"
\end{eqnarray*}
and it obtains the unique identifier from the miner's public key:
\[
uniqueID=getID(publicCert)
\]
Note that the unique identifier is usually contained on Serial Number
of the certificate, or the Subject Alternative Name extension under
different OIDs, depending on the country.
\item Anonymously contact the attested encrypted database of the permissionless
blockchain to register the generated pseudonym and the new public
key $pk$: the uniqueID is also included using the temporary encrypted
key, but it won't be revealed to the host computer of the blockchain
node because it will only be decrypted within the attested enclave.
\item The blockchain's attested code checks within its encrypted database
that the uniqueID has never been included: then, it proceeds to store
the encrypted uniqueID (i.e., this time with a database secret key
that only resides within the enclaves), the generated pseudonym and
the new public key $pk$.
\item Then, the encrypted database's attested code contacts the permissionless
blockchain to register the generated pseudonym and its new public
key $pk$.
\end{enumerate}
\textbf{Taking offline registrations from a permissionless blockchain.
}This protocol takes offline a pseudonym and its associated public
key $pk$ from a permissionless blockchain. To take offline an identity
from a permissionless blockchain, miners must re-run the previous
protocol to prove that the client is able to re-generate the pseudonym
with the same certificate, but this time appending the string ``OFF''
to the pseudonym.

The registering encrypted database of the permissionless blockchain
verifies that the encrypted uniqueID is included in the database before
taking offline the pseudonym from the permissionless blockchain without
it learning the real identity of the miner.

\subsubsection{Mutual Attestation for ePassports\label{subsec:Mutual-Attestation-ePassports}}

Analogous to the zero-knowledge protocols for X.509 \ref{subsec:Mutual-Attestation-X509},
but now considering the specific details of ePassports \ref{subsec:Electronic-Passport},
which usually contain a unique keypair with the public key on Data
Group 15 and the private key hidden within the chip: the Active Authentication
protocol can be used to sign random challenges that can be verified
with the corresponding public key. Some ePassports don't feature Active
Authentication, nonetheless a modified version of the following protocols
could still be executed (see subsection \ref{subsec:attAbsence-of-AA}).

\textbf{Anonymous miner registration of a new public key on a permissionless
blockchain.} This protocol generates a unique pseudonym for each miner,
with a new public key linked to the blockchain-specific pseudonym
that is in turn uniquely linked to the citizen's ePassport: the mutual
attestation between the parties guarantees the correctness of the
execution of both parties.

The following are the steps to the protocol:
\begin{enumerate}
\item The client locally generates a signature secret using its secret key:
\begin{eqnarray*}
signatureSecret & = & \mbox{Sign}\left(secretKey_{publicCert},\right.\\
 &  & \left."\mbox{PREFIXED\_COMMON\_STRING}"\right)
\end{eqnarray*}
\item Mutual attestation between the authenticating client and the blockchain:
the attestation is anonymous thanks to the use of unlinkable signatures
(Enhanced Privacy ID -EPID-), and both parties obtain a temporary
secret key to encrypt their communications.
\item Client's attested code checks the validity of the Data Security Object
of the miner's ePassport:

\begin{enumerate}
\item Load the public key of the Country Signing Certificate from a trusted
source\cite{icaoPKD} and the Document Signing Certificate from the
ePassport.
\item Hash all the Data Groups and check their equivalence to the Data Security
Object.
\item Verify the signature of the Data Security Object using the Document
Signing Certificate.
\item Verify the signature of the Document Signing Certificate using the
Country Signing Certificate.
\item Check the time validity of the certificates.
\end{enumerate}
\item If the previous step concluded satisfactorily, then the client's attested
code verifies the $signatureSecret$ using the miner's public key
certificate $publicCert$ because the $signatureSecret$ is calculated
offline by the ePassport's chip:
\[
\mbox{PKCS\_Verify}\left(publicCert,signatureSecret\right)
\]
The publicCert is taken from the Data Group 15.
\item If the previous step concluded satisfactorily, then the client's attested
code creates a deterministic public/secret key pair based on a secret
passphrase:
\[
pk,sk=\mbox{Det\_KeyPairGen}\left(KDF\left(passphrase,hash(publicCert)\right)\right)
\]
The generation algorithm must be deterministic because the ePassport
is unable to store them and/or the miner may lose them (i.e., as in
deterministic wallets). KDF is a password-based key derivation function
(e.g., PBKDF2).
\item The client's attested code generates a deterministic pseudonym using
the blockchain identifier:
\begin{eqnarray*}
pseudonym & = & Hash\left(signatureSecret||BlockchainIdentifier||uniqueID\right)\\
 &  & ||"\mbox{REG}"
\end{eqnarray*}
and it obtains the unique identifier from the ePassport:
\[
uniqueID=getID(DataGroups)
\]
Note that the unique identifier is usually contained on the Data Element
“Document Number” of the Data Group 1: as it's legally valid for the
same person to own multiple passports with different Document Numbers,
some countries include a unique ``Personal Number'' on the Data
Group 11. Sign is the Active Authentication protocol executed within
the ePassport's chip, a deterministic signing algorithm executed on
a prefixed string to obtain a unique, non-predictable secret based
on the certificate's owner.
\item Anonymously contact the attested encrypted database of the permissionless
blockchain to register the generated pseudonym and the new public
key $pk$: the uniqueID is also included using the temporary encrypted
key, but it won't be revealed to the host computer of the blockchain
node because it will only be decrypted within the attested enclave.
\item The blockchain's attested code checks within its encrypted database
that the uniqueID has never been included: then, it proceeds to store
the encrypted uniqueID (i.e., this time with a database secret key
that only resides within the enclaves), the generated pseudonym and
the new public key $pk$.
\item Then, the encrypted database's attested code contacts the permissionless
blockchain to register the generated pseudonym and its new public
key $pk$.
\end{enumerate}
\textbf{Taking offline registrations from a permissionless blockchain.
}This protocol takes offline a pseudonym and its associated public
key $pk$ from a permissionless blockchain. To take offline an identity
from a permissionless blockchain, miners must re-run the previous
protocol to prove that the client is able to re-generate the pseudonym
with the same certificate, but this time appending the string ``OFF''
to the pseudonym.

The registering encrypted database of the permissionless blockchain
verifies that the encrypted uniqueID is included in the database before
taking offline the pseudonym from the permissionless blockchain without
it learning the real identity of the miner.

\subsubsection{Mapping to goals}

The previous protocols achieve the security goals:
\begin{enumerate}
\item The registered miner's key to the blockchain \textit{opens, but no
one can shut; he can shut, but no one can open}\cite{isaiah2222}.
Only someone in possession of a valid public certificate can create
a unique miner identity on the open blockchain and destroy it. The
signing and verification operations of steps 1 and 4 are only needed
if it's required to check that the miner is the real owner of the
smartcard/ePassport.
\item Protection against malicious issuers: the passphrase is mandatory,
there's no OCSP checking and the protocol is non-bruteforceable because
it can be rate-limited.
\item Privacy: miner's real identity can't be learned by anyone because
publicCert and uniqueID are keep secret.
\item Unique pseudonyms: from each identity card/ePassport, only one unique
identifier can be generated because there's only one uniqueID per
citizen.
\end{enumerate}
The proposed solution depends on the security of Intel SGX (enclave
and remote attestation protocols): in order to limit the impact of
side-channels attacks on Intel SGX, mining nodes featuring the role
of the Attested Encrypted DB will be restricted to trustworthy nodes.

\subsubsection{Performance Evaluation\label{subsec:Performance-Evaluation}}

\begin{tabular}{|c|c|c|c|c|}
\hline 
\# VMs & Mean Time/Req. & \#Req./Sec & Time/Connections & Total time\tabularnewline
\hline 
\hline 
1 VM & 416 ms & 4.76 & 210 ms & 21 secs\tabularnewline
\hline 
4 VM & 112 ms & 16.5 & 59 ms & 5.9 secs\tabularnewline
\hline 
\end{tabular}

A load testing scenario featuring an Intel Xeon E3-1240 3.5 GHz and
running 1 or 4 virtual machines was performed (with 5 users executing
100 requests per user). Operations like reading and/or signing from
the smartcard were not included in the performance evaluation. The
implementation will be open-sourced.

\subsubsection{Absence of Active Authentication\label{subsec:attAbsence-of-AA}}

Signing using the secret key of the Active Authentication protocol
provides an extra layer of security: it guarantess that the remote
party executing the protocol owns a physical copy of the ePassport
(i.e., it hasn't stolen a copy of the public certificates from others).
However, some ePassports don't feature Active Authentication, requiring
a simplified version of the previous protocols by removing steps 1
and 4.

\subsubsection{Removing Single-Points of Failure}

One of the shortcomings of relying on Intel's Attestation Service
(IAS) is that it becomes a single-point of failure: in practice, Intel
would learn who is performing the attestation. For a public, permissionless
blockchain it would be preferable to remove this trusted third party:
to solve this problem, OPERA\cite{operaSGX} provides the first open,
privacy-preserving attestation service to substitute Intel's Attestation
Service.

\subsubsection{Substituting EPID for DCAP}

Since Intel is deprecating EPID (Enhanced Privacy ID) in favour of
DCAP (Data Center Attestation Primitives \cite{intelDCAP,foundationsDCAP}),
a new re-implementation of the previously described mutual attestation
protocol \ref{subsec:Mutual-Attestation-ePassports} has been carried
out using Occlum\cite{occlum}: additionally, one of the benefits
of DCAP is that it removes Intel as a trusted third party as it doesn't
use Intel's Attestation Service. 

Remote biometric authentication using smartphones has also been improved
with remote attestation and DCAP, as described below \ref{subsec:Remote-attestation-smartphones}.

\subsection{Worldwide Coverage and Distribution\label{subsec:Worldwide-Coverage-and}}

\begin{figure}[H]
\includegraphics[scale=0.3]{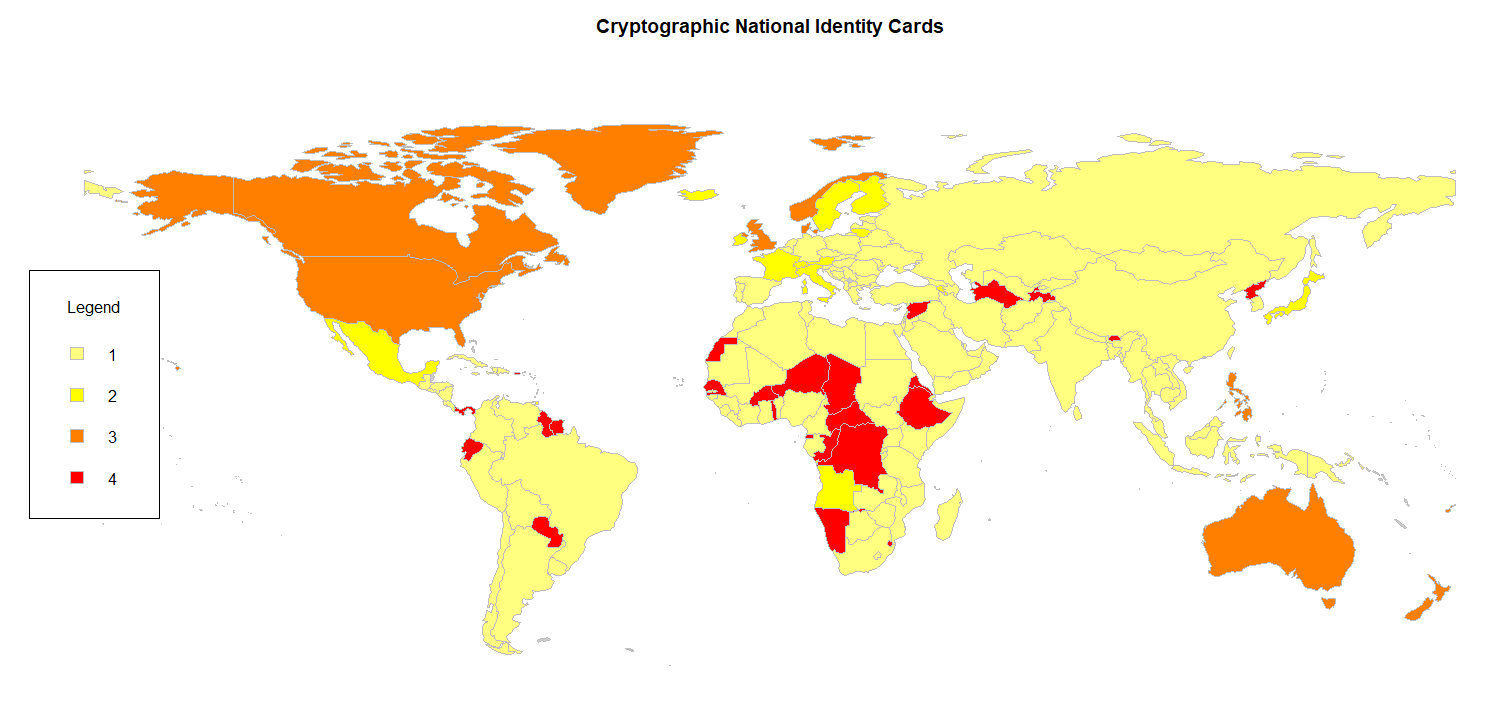}

\caption{\label{fig:Legend:-(1)-National}Legend: (1) National identity card
is a mandatory smartcard; (2) National identity card is a voluntary
smartcard; (3) No national identity card, but cryptographic identification
is possible using an ePassport, driving license and/or health card;
(4) Non-digital identity card.}
\end{figure}

Fortunately, there is a unique cryptographic identifier for most people
in the world: figure \ref{fig:Legend:-(1)-National} shows a worldmap
of the distribution of national identity cards. For some countries,
there is no national identity card -code 3-, but some other unique
cryptographic identifier is available (e.g., ePassport\cite{cryptoeprint:2009:200}
and/or biometric passports as in figure \ref{fig:Availability-of-biometric},
social security card, driving license and/or health card). Transforming
these unique identifiers into anonymous credentials enables the unique
identification of individuals in a permissionless blockchain without
revealing their true identities, making them indistinguishable: that
is, authentication is not only anonymous but permissionless since
there is no need to be pre-invited. Please note the enormous cost
savings resulting from this approach compared to other anonymous credential\cite{cryptoeprint:2013:622,coconut,DBLP:conf/ccs/CamenischDD17}
proposals that would require re-issuing new credentials: for example,
consider that the UK's national identity scheme was estimated at £5.4bn\cite{homeOfficeCostReport}.

In some cases, an individual could obtain multiple cryptographic identifiers
(e.g., multiple nationalities), but their number would still be limited
and certainly less than the number of mining nodes that could be spawned
on PoW permissionless blockchains. Additionally, the true identities
provided by national identity cards could be used for other purposes,
such as non-anonymous accounts identified by their legal identities.

\noindent \begin{center}
\begin{figure}[H]
\includegraphics[scale=0.14]{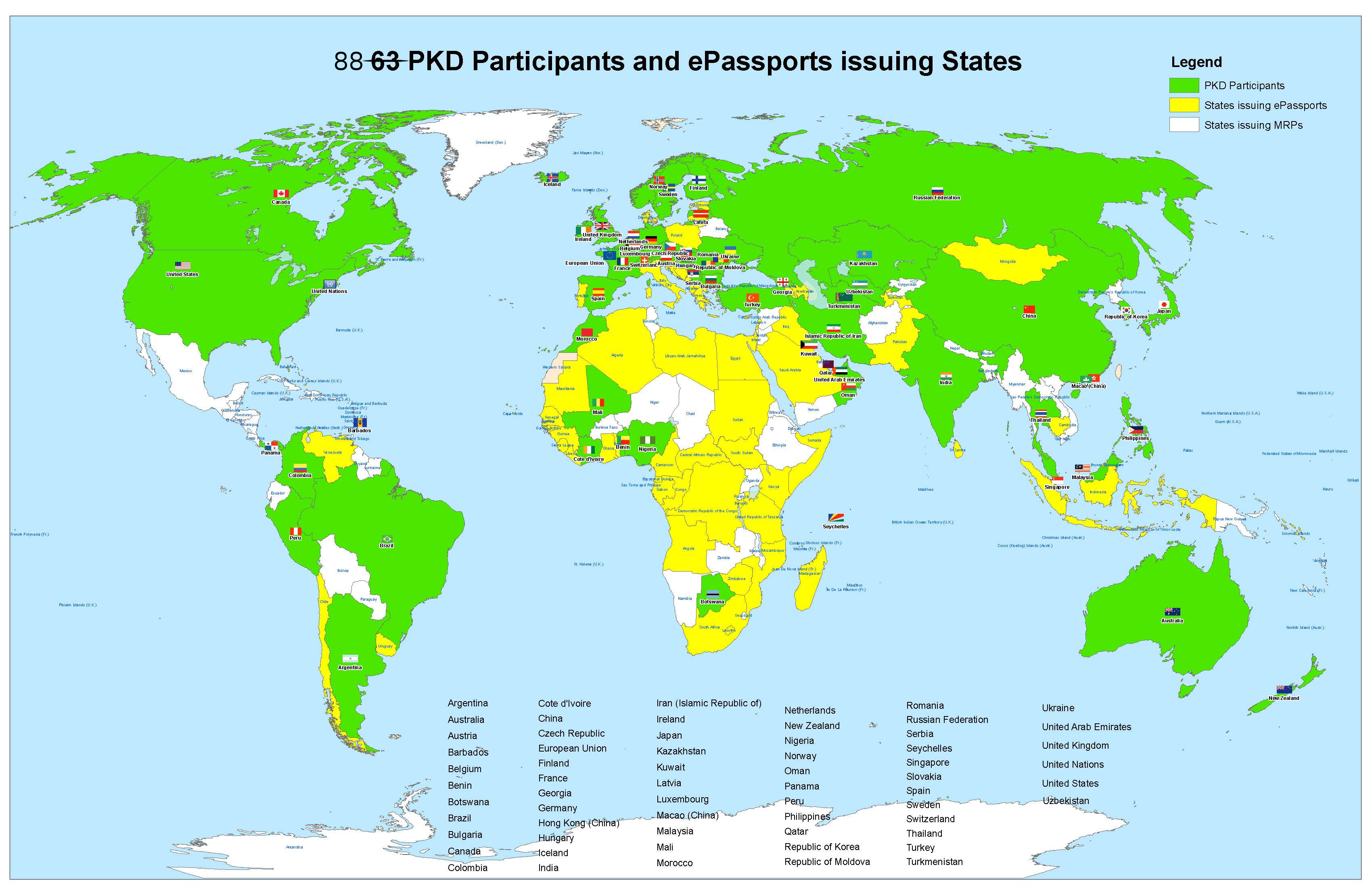}

\caption{\label{fig:Availability-of-biometric}Availability of biometric passports.
Source (ICAO, 2019)}
\end{figure}
\par\end{center}

\noindent \begin{flushleft}
\begin{figure}[H]
\centering{}\includegraphics[scale=0.3]{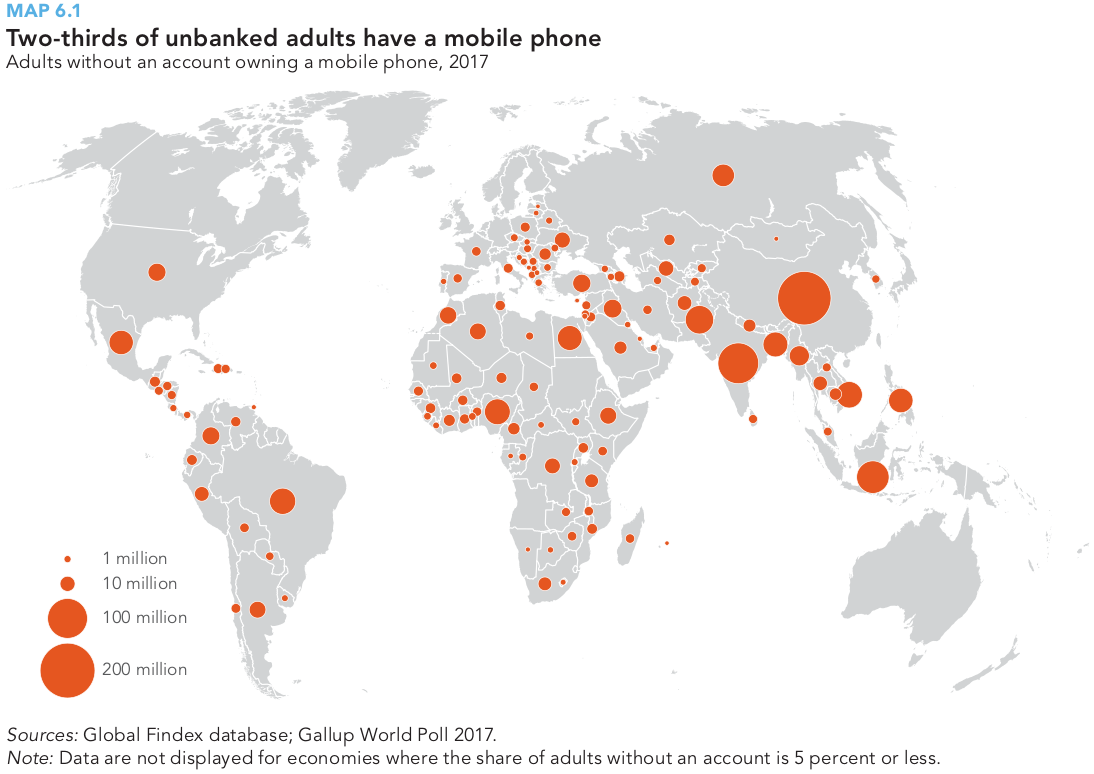}
\end{figure}
\par\end{flushleft}

\subsubsection{eSIM's Public Key Infrastructure}

Latest specifications of SIM cards determine that SIM's identity and
data can be downloaded and remotely provisioned to devices\cite{eSIMwhitepaper}:
instead of the traditional SIM card, there is an embedded SIM (i.e.,
eUICC\cite{eUICCtechnical}) that can store multiple SIM profiles
containing the operator and subscriber data that would be stored on
a traditional SIM card (e.g., IMSI, ICCID, ...).

A novel public key infrastructure has been created in order to protect
the distribution of these new eSIM profiles\cite{gsmaPKI}: every
certified eSIM is signed by its certified manufacturer, with a certificate
that is itself signed by the GSMA root certificate issuer\cite{gsmaRootCertificateIssuer}.
Network operators must also get certified and obtain certificates
for their Subscription Management roles.

The eSIM's PKI provides an alternative identification system for users
where national identity cards and/or ePassports are difficult to obtain,
as they must be unique and non-anonymous (4.13 and 4.1.5\cite{gsmaPKI}),
but only when the mobile operator's KYC processes can be considered
trustworthy.

\newpage{}

\subsubsection{Combining with Non-Zero-Knowledge Authentication}
\noindent \begin{flushleft}
\begin{figure}[H]
\begin{centering}
\includegraphics[scale=0.7]{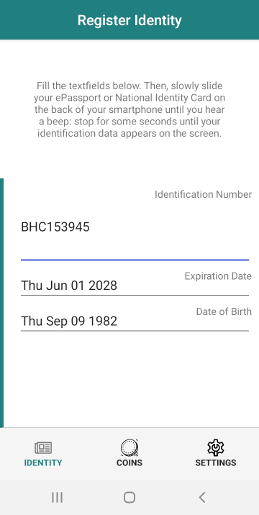}\includegraphics[scale=0.7]{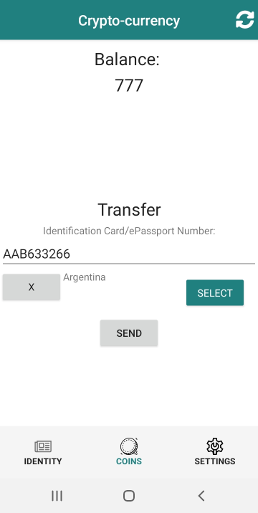}
\par\end{centering}
\caption{\label{fig:Mobile Authentication}Mobile App using Non-Zero-Knowledge
Authentication}
\end{figure}
\par\end{flushleft}

The use of zero-knowledge techniques to authenticate miners doesn't
preclude non-zero-knowledge authentication on the same blockchain:
figure \ref{fig:Mobile Authentication} shows a mobile app\cite{calctopiaApp}
using BAC authentication to read a National Identity Card and/or ePassport
(left), and then transferring to another ePassport (right). All national
identifiers are publicly addressable.

\subsubsection{Remote attestation for smartphones\label{subsec:Remote-attestation-smartphones}}

Remote biometric authentication using smartphones (Android/iOS) has
been improved with remote attestation \cite{feido} based on Intel
SGX Data Center Attestation Primitives:
\begin{itemize}
\item much faster and less memory consumption than solutions based on pure
zero-knowledge proofs, thus even low-end smartphones are supported:
note that pure zero-knowledge solutions may require gigabytes of RAM
to verify the full certificate chain, thus excluding low-end smartphones
(i.e., inclusion is the primary goal of this authentication solution)
\item support for all the complete ciphersuite used in biometric passports/electronic
identity cards: the signatures schemes used in blockchains for zero-knowledge
proofs (e.g., BLS, Secp256k1) are different from the ones used in
government IDs (Brainpool, ED448, ...), thus zero-knowledge implementations
do not even exist yet. Worse still, ICAO standards \cite{icaoDoc9303part10,icaoDoc9303part11}
leave open the use of new elliptic curves described by their parameters,
difficulting the implementation of pure zero-knowledge solutions.
\item better updatability: new techniques for anti-money laundering and/or
sanction screening can be easily implemented by upgrading the server-side
enclaves, while pure zero-knowledge solutions are constrained to previous
data structures
\end{itemize}
Note that the use of remote attestation doesn't exclude the possibility
of using zero-knowledge proofs: both techniques can be combined together,
but relegating zero-knowledge proofs to less demanding cases of selective
disclosure of credentials.

\subsection{Circumventing the Impossibility of Full Decentralization\label{subsec:Circumventing-the-Impossibility}}

Most blockchains using PKIs are consortium blockchains, thus it has
become widespread that they always are permissioned and centralised.
However, the term ``permissionless'' literally means ``without
requiring permission'' (i.e., to access, to join, ...), thus a blockchain
with a PKI could be permissionless if it accepts any self-signed certificate
(i.e., a behaviour conceptually equivalent to Bitcoin), or any certificate
from any government in the world as described in the previous subsection
\ref{subsec:Worldwide-Coverage-and}. In the same way, a blockchain
using PKIs doesn't imply that its control has become centralised,
it means that it accepts identities from said PKIs as described in
this paper: actually, decentralization in the blockchain context strictly
means that the network and the mining are distributed in a large number
of nodes, thus unrelated to authentication.

A recent publication\cite{impossibleDecentralization} proves that
it's impossible for blockchains to be fully decentralised without
real identity management (e.g., PoW, PoS and DPoS) because they cannot
have a positive Sybil cost, defined as the additional cost that a
player should pay to run multiple time nodes compared to the total
cost of when those nodes are run by different players. To reflect
the level of decentralization, they introduce the following definition:
\begin{defn}
($\left(m,\epsilon,\delta\right)$-Decentralization)\cite{impossibleDecentralization}.
For $1\leq m$, $0\leq\epsilon$ and $0\leq\delta\leq100$, a system
is $\left(m,\epsilon,\delta\right)$-decentralized if it satisfies
the following:
\end{defn}

\begin{enumerate}
\item The size of the set of players running nodes in the consensus protocol,
$P$, is not less than $m$ (i.e., $\left|P\right|\geq m$).
\item Define $EP_{p_{i}}$ as the effective power of player $p_{i}$ as
$\sum_{n_{i}\in N_{p_{i}}}\alpha_{n_{i}}$ where $N$ is the set of
all nodes in the consensus protocol and $\alpha_{p_{i}}$ is the resource
power of player $p_{i}$. The ratio between the effective power of
the richest player, $EP_{max}$, and the $\delta-th$ percentile player,
$EP_{\delta}$, is less than or equal to $1+\epsilon$ (i.e., $\left(EP_{max}/EP_{delta}\right)\leq1+\epsilon$).
\end{enumerate}
Ideally, the number $m$ should be as high as possible (i.e., too
many players do not combine into one node); and for the most resourceful
and the $\delta$-th percentile player running nodes, the gap between
their effective power is small. Therefore, full decentralization is
represented by $\left(m,0,0\right)$ for sufficiently large $m$.
\begin{defn}
(Sufficient Conditions for Fully Decentralized Systems)\cite{impossibleDecentralization}.
The four following conditions are sufficient to reach $\left(m,\epsilon,\delta\right)$-decentralization
with probability 1.
\end{defn}

\begin{enumerate}
\item (Give Rewards (GR-$m$)). Nodes with any resource power earn rewards.
\item (Non-delegation (ND-$m$)). It is not more profitable for too many
players to delegate their resource power to fewer participants than
to directly run their own nodes.
\item (No Sybil nodes (NS-$\delta$)). It is not more profitable for a participant
with above the $\delta$-th percentile effective power to run multiple
nodes than to run one node.
\item (Even Distribution (ED-($\epsilon,\delta$)). The ratio between the
resource power of the richest and the $\delta$-th percentile nodes
converges in probability to a value less than $\text{1+\ensuremath{\epsilon}}$.
\end{enumerate}
\begin{thm}
For any initial state, a system satisfying GR-$m$, ND-$m$, NS-$\delta$,
and ED-$\left(\epsilon,\delta\right)$ converges in probability to
$\left(m,\epsilon,\delta\right)$-decentralization. \cite{impossibleDecentralization}
\end{thm}

As it should be obvious by now, a blockchain using zk-PoI with strong
identities from trusted public certificates (e.g., national identity
cards and/or ePassports \ref{subsec:Worldwide-Coverage-and}) as described
in this paper is the perfect candidate to achieve a fully decentralized
blockchain.
\begin{thm}
A blockchain using zk-PoI with strong identities from trusted public
certificates (e.g., national identity cards and/or ePassports \ref{subsec:Worldwide-Coverage-and})
reaches $\left(m,\epsilon,\delta\right)$-decentralization with probability
1.
\end{thm}

\begin{skproof}A blockchain using zk-PoI with strong identities from trusted public certificates effectively limits the number of mining nodes to one per individual (ND-$m$), independently of how resourceful they are (NS-$\delta$, ED-($\epsilon,\delta$)), while keeping membership open to everyone (i.e., achieves a large number of participants (GR-$m$)). The presence of strong identities allows positive Sybil costs, thus the fulfillment of the sufficient conditions for fully decentralized systems\cite{impossibleDecentralization}.
\end{skproof}

Preventing delegation (ND-$m$) is the most difficult condition to
meet:
\begin{itemize}
\item market-enforced: richest participants could buy rights-of-use of others'
identities, but the market value of said identities (e.g., the Net
Present Value of future profits obtained from the exploitation of
said identities by their real owners) should wipe away almost all
the profits from these exchanges.
\item strictly-enforced: miners' software could frequently check for the
presence of the physical trusted public certificate (e.g., national
identity cards and/or ePassports) and/or require them when transferring
funds out of their accounts.
\end{itemize}
A posterior revision of the paper\cite{impossibleDecentralizationv2}
introduces new definitions that try to emphasize that Trusted Third
Parties (i.e., Certificate Authorities) shouldn't be used in decentralized
blockchains: as described in this paper\ref{subsec:Worldwide-Coverage-and},
using 400 CAs of national identity cards and ePassports from over
the world is still being decentralized according to the original definition
of decentralization, and certainly much more decentralized than Bitcoin/Ethereum
that concentrate >50\% of their hashrate in 4-3 entities\cite{decentralizationBitCoinEthereum}.
It's very important not to mix concepts and misattribute qualities
to different concepts:
\begin{itemize}
\item permissionless doesn't imply without identification
\item a decentralized consensus protocol doesn't imply that it can't use
identification from TTPs: recent consensus protocols decouple Sybil-resistance
from the consensus mechanism\ref{subsec:Modern-Consensus-based-on}
\end{itemize}
Thus, permissionless, decentralized and identification(-less) are
different qualities that shouldn't be intermixed. The results of this
paper hold in the \textit{trusted} decentralized setting, while the
results of \cite{impossibleDecentralizationv2} hold in the \textit{trustless}
decentralised setting.

\subsection{Resistance against Dark DAOs}

Dark DAOs\cite{darkDaos} appear as a consequence of permissionless
blockchains where users can create their own multiple identities and
there's no attributability of the actions.
\begin{enumerate}
\item When using real-world identities, it's possible to establish the identity
of the parties running the Dark DAO that are committing frauds (attributability)
or at least, their pseudonyms: therefore, it's possible to punish
them.
\item To prevent Dark DAOs buying real-world identities, a smart contract
can be setup that pays a reward for denouncing the promoters of the
fraud: the whistleblowers would be paid a multiple of what they would
get from the defrauders, thus making denunciation the preferred option.
Then defrauders would be banned as in step 1.
\end{enumerate}

\subsection{Resistance against Collusion and other Attacks}

In this sub-section, we consider different avenues for attack and
provide detailed defense mechanisms:
\begin{enumerate}
\item Corrupt root certificate authorities
\item Attacks against consensus protocols
\item Resistance against collusion
\item On achieving collusion-freeness
\end{enumerate}

\subsubsection{Corrupt Root Certificate Authorities}

Corrupt countries may be tempted to create fake identities or frequently
renovate existing ones: these countries can be easily banned out by
removing them from the list of valid authorities (i.e., root X.509
certificate and/or Country Signing Certificate). Bounties in cryptocurrency
could be offered for whistleblowing any corrupt attack against the
long-term existence of the blockchain.

\subsubsection{Attacks against Consensus Protocols}

Modern consensus protocols based on the cryptographically secure random
choice of the leader (e.g., \cite{dfinityConsensus,cryptoeprint:2017:406})
detect cheating by monitoring changes to the chain quality. The following
table gathers cheating events for different consensus algorithms that
could be detected and punished:

\noindent \begin{center}
\begin{tabular}{|c|>{\raggedright}p{7.5cm}|}
\hline 
Protocol & Cheater detection\tabularnewline
\hline 
\hline 
\multirow{4}{*}{Dfinity\cite{dfinityConsensus}} & Equivocation: multiple blocks for same round with same rank.\tabularnewline
\cline{2-2} 
 & Equivocation: multiple blocks with the highest priority.\tabularnewline
\cline{2-2} 
 & All the blocks must be timely published.\tabularnewline
\cline{2-2} 
 & All the notarizations must be timely published within one round.\tabularnewline
\hline 
\multirow{3}{*}{OmniLedger\cite{cryptoeprint:2017:406,cryptoeprint:2019:676}} & Core validators can detect rogue validators.\tabularnewline
\cline{2-2} 
 & Withholders can be detected after multiple consecutive rounds.\tabularnewline
\cline{2-2} 
 & 5>= failed RandHound views from a rogue validator.\tabularnewline
\hline 
\end{tabular}
\par\end{center}

\subsubsection{Resistance against Collusion}

Consensus protocols already provide collusion-tolerance by design:
an adversary controlling a high number of nodes, or equivalently all
said nodes colluding for the same attack, must confront the difficulties
introduced by shard re-assignment at the beginning of every new epoch.
For the case of OmniLedger\cite{cryptoeprint:2017:406}, the security
of the validator assignment mechanism can be modeled as a random sampling
problem using the binomial distribution,
\[
P\left[X\leq\left\lfloor \frac{n}{3}\right\rfloor \right]=\sum_{k=0}^{n}\left(\begin{array}{c}
n\\
k
\end{array}\right)m^{k}\left(1-m\right)^{n-k}
\]
assuming that each shard has less than $\left\lfloor \frac{n}{3}\right\rfloor $
malicious validators. Then, the failure rate of an epoch is the union
bound of the failures rates of individual shards, each one calculated
as the cumulative distribution over the shard size $n$, with $X$
being the random variable that represents the number of times we pick
a malicious node. An upper bound of the epoch failure event, $X_{E}$,
is calculated as:
\[
P\left[X_{E}\right]\leq\sum_{k=0}^{l}\frac{1}{4^{k}}\cdot n\cdot P\left[X_{S}\right]
\]
where $l$ is the number of consecutive views the adversary controls,
$n$ is the number of shards and $P\left[X_{S}\right]$ is the failure
rate of one individual shard. Finally, for $l\rightarrow\infty$,
we obtain
\[
P\left[X_{E}\right]\leq\frac{4}{3}\cdot n\cdot P\left[X_{S}\right]
\]

\subsubsection{On Achieving Collusion-Freeness}

Start noticing that collusion-freeness is not about preventing malicious
behaviour, only preventing that malicious players act as independently
of each other as possible. Following a previous seminal work\cite{collusionFreeProtocols},
collusion-freeness can only be obtained under very stringent conditions:
(a) the game must be finite; (b) the game must be publicly observable;
and (c) the use of private channels at the beginning of the game is
essential, but forbidden during the execution of the protocol. Although
blockchains are publicly observable, they are also an infinite game
where parties can freely communicate between them using private channels
at any time: therefore, collusion-freeness is impossible in the sense
of \cite{collusionFreeProtocols}.

Fortunately, there is a way to get around this impossibility result:
forbid malicious/Byzantine behaviours requiring the use of mutual
attestation for all the nodes, thus precluding any deviation from
the original protocol.
\begin{conjecture}
If mutual attestation is required for all the nodes, any infinite,
partial-information blockchain game with publicly observable actions
has a collusion-free protocol.
\end{conjecture}

As mutual attestation is already required for zk-PoI \ref{subsec:Detailed-Authentication-Protocol-Remote-Attestation},
we would only be extending its use for the rest of the blockchain
protocol.

\section{Incentive Compatible and Strictly Dominant Cryptocurrencies}

The success of cryptocurrencies is better explained by their incentive
mechanisms rather than their consensus algorithms: a cryptocurrency
with poor incentives (e.g., a cryptocurrency not awarding coins to
miners) will not achieve any success; conversely, better incentives
and much more inefficient consensus algorithm could still find some
success.

Much research has been focused on conceiving better consensus algorithms
for decentralised systems and cryptocurrencies\cite{cryptoeprint:2016:918,cryptoeprint:2016:919,cryptoeprint:2017:406,cryptoeprint:2017:454,DBLP:journals/corr/Kokoris-KogiasJ16,dfinityConsensus}:
unfortunately, obtaining consensus mechanisms with better incentives
and economic properties is an area that is lacking much research,
and the combination of all the game-theoretic results contained in
this section fills this gap for the sake of achieving a \textit{focal
point} (i.e., Schelling point\cite{schellingConflict}) in the multi-equilibria
market of cryptocurrencies. Thus, a selective advantage is introduced
by design over all the other cryptocurrencies, in explicit violation
of the neutral model of evolution\cite{2017arXiv170505334E} in order
to obtain an incentive compatible and strictly dominant cryptocurrency.

\subsection{Incentive-Compatible Cryptocurrency\label{subsec:Incentive-Compatible-Cryptocurre}}

Shard-based consensus protocols have been recently introduced in order
to increase the scalability and transaction throughput of public permissionless
blockchains: however, the study of the strategic behaviour of rational
miners within shard-based blockchains is very recent. Unlike Bitcoin,
that grants all rewards to the most recent miner, block rewards and
transactions fees must be proportionally shared between all the members
of the sharding committee\cite{DBLP:journals/corr/Kokoris-KogiasJ16},
and this includes incentives to remain live during all the lifecycle
of the consensus protocol. Even so, existing sharding proposals\cite{cryptoeprint:2017:406,cryptoeprint:2018:460}
remain silent on how miners will be rewarded to enforce their good
behaviour: as it's evident, if all miners are equally rewarded without
detailed consideration of their efforts, rational players will \textit{free-ride}
on the efforts of others. 

One significant difference introduced in this paper with respect to
other shard-based consensus protocols is the use of Zero-Knowledge
Proof-of-Identity as the Sybil-resistance mechanism: as we will see
in the following sections, it's a significant novelty because solving
Proof-of-Work puzzles is the most computationally expensive activity
of consensus protocols, thus it's no longer dominated by computational
costs. This makes the necessity for an incentive-compatible protocol
even more acute: the preferred rational miner's strategy is to execute
the Proof-of-Work of the initial phase of the protocol for each epoch
and to refrain from the transaction verification and consensus of
subsequent phases of the protocol, but still selfishly claim the rewards
as if they had participated. The substitution of costly PoW for cheap
Zero-Knowledge Proof-of-Identity only increases the attractiveness
of this rational strategy, that can only be counteracted by using
an incentive-compatible protocol.

\subsubsection{A Nash Equilibrium for a Cryptocurrency on a Shard-Based Blockchain}

This section is based on a stylised version of a recent game-theoretic
model\cite{2018arXiv180907307M}, taking into consideration that there
is no cost associated with committee formation to enter each shard
since we are using Zero-Knowledge Proof-of-Identity, and not costly
Proof-of-Work: instead, a penalty $p$ is imposed to defective and/or
cheating miners. The following is a list of symbols:

\begin{tabular}{|c|c|}
\hline 
Symbol & Definition\tabularnewline
\hline 
\hline 
$k$ & Number of shards\tabularnewline
\hline 
$N$ & Number of miners\tabularnewline
\hline 
$x_{i}^{j}$ & Vector of received transactions by miner $i$ in shard $j$\tabularnewline
\hline 
$y^{j}$ & Vector of transactions submitted by shard $j$ to blockchain\tabularnewline
\hline 
$c$ & Minimum number of miners in each committee\tabularnewline
\hline 
$\tau$ & Required number of miners in shard for consensus\tabularnewline
\hline 
$r$ & Benefit for each transaction\tabularnewline
\hline 
$b_{i}$ & Benefit of miner $i$ after adding the block\tabularnewline
\hline 
$c_{i}^{t}$ & Total cost of computation for miner $i$\tabularnewline
\hline 
$c^{o}$ & Total optional costs in each epoch\tabularnewline
\hline 
$c^{v}$ & Cost of transaction verification\tabularnewline
\hline 
$c^{f}$ & Fixed costs in optional cost\tabularnewline
\hline 
$p$ & Penalty cost\tabularnewline
\hline 
$BR$ & Block Reward\tabularnewline
\hline 
$l_{j}$ & Number of cooperative miners in each shard\tabularnewline
\hline 
$L$ & Total number of cooperative miners in all shards\tabularnewline
\hline 
$C_{j}^{l_{j}}$ & Set of all cooperative miners in shard $j$\tabularnewline
\hline 
$D_{j}^{n-l_{j}}$ & Set of all defective miners in shard $j$\tabularnewline
\hline 
$C^{L}$ & Set of all cooperative miners\tabularnewline
\hline 
$D^{N-L}$ & Set of all defective miners\tabularnewline
\hline 
$s^{r}$ & Signed receipt of a transaction\tabularnewline
\hline 
\end{tabular}\\

Let $\mathbb{G}$ denote the shard-based blockchain game, defined
as a triplet $\left(P,S,U\right)$ where $P=\left\{ P_{i}\right\} _{i=1}^{N}$
is the set of players, $S=\left\{ C,D\right\} $ is the set of strategies
(Cooperate $C$, or Defect $D$) and $U$ is the set of payoff values.
Each miner receives a reward if and only it has already cooperated
with other miners within the shard, the payoff of cooperative miners
in set $C^{l_{j}}$ is 

\begin{equation}
u_{i}\left(C\right)=\frac{BR}{kl_{j}}+\frac{r\left|y^{j}\right|}{l_{j}}-\left(c^{f}+\left|x_{i}^{j}\right|c^{v}\right)\label{eq:5.1}
\end{equation}
We assume that the block reward $BR$ is uniformly distributed among
shards and each cooperative miner can receive a share of it. A miner
might be cooperative but all other miners may agree on a vector of
transactions $y^{j}$ that is different from his own vector of transactions
$x_{i}^{j}$ (i.e., $\left|x_{i}^{j}\right|\neq\left|y^{j}\right|$):
nonetheless, transaction rewards are uniformly distributed among all
cooperative miners in each shard, proportional to all the transactions
submitted to the blockchain by each shard.

The defective miners' payoff can be calculated as
\[
u_{i}^{D}=-p^{m}
\]
because the defective miners will have to pay a penalty and they will
not receive any benefit (and it doesn't incur in any mandatory cost
such as solving PoW puzzles because we use cheap Zero-Knowledge Proof-of-Identity). 

There exists a cooperative Nash equilibrium profile in game $\mathbb{G}$
under the following conditions:
\begin{thm}
\label{thm:UniqueNashEqulibriumCryptocurrency}Let $C_{j}^{l_{j}}$
and $D_{j}^{m-l_{j}}$ denote the sets of $l_{j}$ cooperating miners
and $n-l_{j}$ defecting miners inside each shard $j$ with $n$ miners,
respectively. $\left(C^{L},D^{N-L}\right)$ represents a Nash equilibrium
profile in each epoch of game $\mathbb{G}$, if the following conditions
are satisfied:

\begin{enumerate}
\item In all shards $j$, $l_{j}\geq\tau$.
\item If for a given miner $P_{i}$ in shard $j$, with $x_{i}^{j}=y^{j}$,
then the number of transactions $\left|x_{i}^{j}\right|$ must be
greater than
\[
\theta_{c}^{1}=\frac{c^{f}-\frac{BR}{kl_{j}}+p}{\nicefrac{r}{l_{j}}-c^{v}}
\]
\item If for a given miner $P_{i}$ in shard $j$, with $x_{i}^{j}\neq y^{j}$,
then the number of transactions $\left|x_{i}^{j}\right|$must be smaller
than
\[
\theta_{c}^{2}=\frac{\frac{BR}{kl_{j}}+\frac{r\left|y^{j}\right|}{l_{j}}-c^{f}-p}{c^{v}}
\]
\end{enumerate}
\end{thm}

\begin{proof}
The first condition $l_{j}\geq\tau$ (i.e., the number of cooperative
miners must be greater than $\tau$) must hold so that cooperative
miners will receive benefits for transactions and block rewards.

Let $l_{j}^{*}$ be the largest set of cooperative miners in each
shard, where no miner in $D_{j}^{n-l_{j}}$ can join $C_{j}^{l_{j}}$
to increase its payoff: if miner $P_{i}^{j}$ is among the set of
cooperative miners where $x_{i}^{j}=y^{j}$, then it would not unilaterally
deviate from cooperation if:
\[
\frac{BR}{kl_{j}}+\frac{r\left|x_{i}^{j}\right|}{l_{j}}-\left(c^{f}+\left|x_{i}^{j}\right|c^{v}\right)\geq-p
\]
which shows that $x_{i}^{j}\geq\theta_{c}^{1}$, whereas in the second
condition, 
\[
\theta_{c}^{1}=\frac{c^{f}-\frac{BR}{kl_{j}}-p}{\nicefrac{r}{l_{j}}-c^{v}}
\]
But if $P_{i}^{j}$ is among the cooperators whose vector of transactions
is different from the output of the shard, $x_{i}^{j}\neq y^{j}$,
then it would not deviate from cooperation if:
\[
\frac{BR}{kl_{j}}+\frac{r\left|y^{j}\right|}{l_{j}}-\left(c^{f}+\left|x_{i}^{j}\right|c^{v}\right)\geq-p
\]
which shows that $x_{i}^{j}<\theta_{c}^{2}$, whereas in the third
condition,
\[
\theta_{c}^{2}=\frac{\frac{BR}{kl_{j}}+\frac{r\left|y^{j}\right|}{l_{j}}-c^{f}-p}{c^{v}}
\]
Then if $l_{j}^{\text{*}}$ represents the largest set of cooperative
miners in each shard, then $\left(C^{L},D^{N-L}\right)$ would be
the unique cooperative Nash equilibrium of the game $\mathbb{G}$.
\end{proof}
As can be understood from the proof, cooperative miners have less
incentive to cooperate when: 1) the number of participants $N$ increases;
2) the optional costs of computation increase ($c^{f}$ is in the
numerator or $c^{v}$ in denominator of $\theta_{C}$); 3) or in general,
when the number of transactions is not large enough compared to a
fixed threshold.

\subsubsection{Incentive-Compatible Cryptocurrency on a Shard-Based Coordinated
Blockchain}

In order to increase the incentives to cooperate rather than defect,
an incentive-compatible protocol enforcing cooperation based on the
previously presented Nash equilibrium is introduced here \ref{alg:Incentive-Compatible-Protocol}:
all miners should disclose their list of transactions to a coordinator,
who then announces to each miner whether their cooperation would be
in their interests based on being within the maximum subset of miners
with a similar list of transactions (i.e., $x_{i}^{j}=y^{j}$), and
then enforces their cooperation by checking their compliance and rewarding
them properly.\\
\begin{algorithm}
\noindent\fbox{\begin{minipage}[t]{1\columnwidth - 2\fboxsep - 2\fboxrule}%
function ShardTransactionsAssignment \{

~~~$Shard\leftarrow GetShard(epochRandomness,pseudonym,PK)$

~~~$x_{i}\leftarrow ShardTransactions(Shard)$

\}

function NodeSelection \{

~~~$P_{i}$ send $H\left(x_{i}^{j}\right)$ to Coordinator
\begin{flushleft}
~~~if (PresentNode() == Coordinator) \{
\par\end{flushleft}
~~~~~~Receive all $H\left(x_{i}^{j}\right)$

~~~~~~$l_{j}\leftarrow$Max number of miners with common txs.
from list of $H\left(x_{i}^{j}\right)$

~~~~~~if ($l_{j}<\tau$)

~~~~~~~~~return ``All Defective''

~~~~~~else \{

~~~~~~~~~$C_{j}^{l_{j}}$= list of $l_{j}$miners

~~~~~~~~~Calculate $\theta_{c}^{1}$ and $\theta_{c}^{2}$
from \thmref{UniqueNashEqulibriumCryptocurrency}

~~~~~~~~~return $\theta_{c}^{1}$, $\theta_{c}^{2}$ and
$C_{j}^{l_{j}}$

~~~~~~\}

~~~\}

\}

function ShardParticipation \{

~~~if ($P_{i}\in C_{j}^{l_{j}}$ and $\left|x_{i}^{j}\right|\leq\theta_{c}^{1}$)

~~~~~~return Defect

~~~else if ($P_{i}\notin C_{j}^{l_{j}}$ and $\left|x_{i}^{j}\right|\geq\theta_{c}^{2}$)

~~~~~~return Defect

~~~Verify transactions

~~~$y^{j}$=set of verified transactions by remaining cooperative
$P_{i}$

~~~Consensus on verified transactions

~~~Sign BFT agreement result

~~~return signature, agreed block's header

\}

function VerificationAndRewards \{

~~~Verify cooperation of $P_{i}\in C^{L}$ for each shard

~~~Send rewards to cooperative $P_{i}$ using \eqref{5.1}

\}%
\end{minipage}}

\caption{\label{alg:Incentive-Compatible-Protocol}Incentive-Compatible Protocol
on a Coordinated Shard-Based Blockchain}

\end{algorithm}

The protocol proceeds as follows: for the first function (i.e., \textit{ShardTransactionsAssignment}),
each miner receives a list of transactions $x_{i}^{j}$ to verify
based on the epochRandomness and his pseudonymous identity and public
key obtained by the Zero-Knowledge Proof-of-Identity.

For the second function (i.e., \textit{NodeSelection}), all miners
calculate a hash $H\left(x_{i}^{j}\right)$ over their transaction
list and send it to the coordinator. The coordinator finds the subset
with the maximum number of miners with a common transaction list,
thus calculating $\theta_{c}^{1}$, $\theta_{c}^{2}$, $l_{j}$ and
$C_{j}^{l_{j}}$: in each epoch, the coordinator publicly defines
the list of cooperative miners $C_{j}^{l_{j}}$ and defective miners
$D_{j}^{n-l_{j}}$ using on \thmref{UniqueNashEqulibriumCryptocurrency}.

For the third function (i.e., \textit{ShardParticipation}), all the
transactions of each miner are verified and a signed consensus is
reached.

For the fourth function (i.e., \textit{VerificationAndRewards}), the
rewards are shared between the cooperative miners and denied to those
miners in $C_{j}^{l_{j}}$ that didn't cooperate.

\subsubsection{Improved Incentive-Compatible Cryptocurrency on a Shard-Based Blockchain}

Although the role of a coordinator is essential to BFT protocols\cite{DBLP:journals/corr/Kokoris-KogiasJ16},
its expanded functionality in the previous incentive-compatible protocol
\algref{Incentive-Compatible-Protocol} is problematic: it introduces
latency and network costs due to the new obligations to report to
the coordinator; moreover, it creates new opportunities for malicious
miners which may report false $H\left(x_{i}^{j}\right)$ or not follow
coordinator's instruction to cooperate or defect. The next incentive-compatible
protocol significantly improves over the state of the art: the role
of the coordinator is minimised, strengthing the protocol by removing
the previous vulnerabilities and making it resistant to malicious
miners. 

Information propagation\cite{bitcoinInformationPropagation} is an
essential part of any blockchain, and gossiping transactions to neighbouring
miners is one of its key features. In the new incentive-protocol protocol,
we require that any broadcasted/gossiped/propagated transaction gets
acknowledged with a signed receipt to its sender: then, senders must
attach these receipts to the consensus leaders and verification nodes
in order to ease detection of defective and/or cheating miners. \textcolor{black}{Miners
who were gossiped transactions but didn't participate are considered
defective, and not rewarded}. In other words, the signed receipts
serve as snitches that denounce non-cooperative miners thus preventing
that any reward gets assigned to them: at the same time, all miners
are incentivised to participate in the denunciation in order to gain
the rewards of non-cooperative miners and other \textit{free-riders}.
\begin{algorithm}[h]
\noindent\fbox{\begin{minipage}[t]{1\columnwidth - 2\fboxsep - 2\fboxrule}%
function ShardTransactionsAssignment \{

~~~$Shard\leftarrow GetShard(epochRandomness,pseudonym,PK)$

~~~$x_{i}\leftarrow ShardTransactions(Shard)$

\}

function GossipTransaction \{

~~~GossipTransaction()

~~~$s^{r}$= AcknowledgeTransmission()

~~~Store $s^{r}$

\}

function ReceiveTransaction \{

~~~$tx$=ReceiveTransaction()

~~~ReplyTransaction(sign(hash($tx$)))

\}

function ShardParticipation \{

~~~Verify transactions

~~~Collect lists of $s^{r}$ for every $P_{i}$

~~~$y^{j}$=set of verified transactions by remaining cooperative
$P_{i}$

~~~Consensus on verified transactions

~~~Sign BFT agreement result

~~~return signature, agreed block's header

\}

function VerificationAndRewards \{

~~~Verify cooperation of $P_{i}\in C^{L}$ using lists of $s^{r}$

~~~Send rewards to cooperative $P_{i}$ using \eqref{5.1}

\}%
\end{minipage}}

\caption{\label{alg:Improved-Incentive-Compatible-Protocol}Improved Incentive-Compatible
Protocol on a Shard-Based Blockchain}

\end{algorithm}

In order to save bandwidth, note that it's not obligatory to send
the full list of all signed transaction receipts to consensus leaders
and/or verification nodes: only a random subset per each miner should
be enough to catch defective miners.

Additionally, the absence of signed receipts could be used to detect
the need of a change of a consensus leader (i.e., ``view-change'')
in BFT protocols\cite{DBLP:journals/corr/Kokoris-KogiasJ16,cryptoeprint:2017:406}.

\subsection{On Strictly Dominant Cryptocurrencies}

A cryptocurrency using Zero-Knowledge Proof-of-Identity as the Sybil-resistance
mechanism strictly dominates PoW/PoS cryptocurrencies: a miner having
to choose between mining different cryptocurrencies, one with no costs
associated with its Sybil-resistance mechanism and distributing equally
the rewards, and the others using costly PoW/PoS and thus featuring
mining concentration, will always choose the first one. That is, mining
equally distributed cryptocurrencies using Zero-Knowledge Proof-of-Identity
is a dominant strategy; in other words, the strategy of mining Bitcoin
and other similar cryptocurrencies is strictly dominated by the hereby
described cryptocurrency. \textit{Ceteris paribus}, this cryptocurrency
will have better network effects, thus better long-term valuation.

\subsubsection{Strictly Dominant Cryptocurrencies and a Nash Equilibrium\label{subsec:Strictly-Dominant-Cryptocurrenci}}

The intuition behind the preference to mine fully decentralised cryptocurrencies
with the lowest expenditure (i.e., lowest CAPEX/OPEX implies higher
profitability), thus the search for better hash functions\cite{cryptoeprint:2017:1168,cryptoeprint:2016:989,cryptoeprint:2015:430,cryptoeprint:2016:027},
is formally proved here and then applied to the specific case of the
proposed cryptocurrency.
\begin{defn}
(Power-Law Fee-Concentrated (PLFC) cryptocurrency)\label{def:(Power-law-fee-concentrated-cryp}.
A cryptocurrency whose distribution of mining and/or transaction fees
follows a power-law (i.e., a few entities earn most of the fees/rewards),
usually due to the high costs of its Sybil-resistance mechanism.
\end{defn}

\begin{example}
Proof-of-Work cryptocurrencies are Power-Law Fee-Concentrated: 90\%
of the mining power is concentrated in 16 miners in Bitcoin and 11
in Ethereum\cite{decentralizationBitCoinEthereum}.

Proof-of-Stake cryptocurrencies are Power-Law Fee-Concentrated: miners
earn fees proportional to the amount of money at stake, and wealth
is Pareto-concentrated\cite{Pareto2014-PARMOP-2}.
\end{example}

\begin{defn}
(Uniformly-Distributed Capital-Efficient (UDCE) cryptocurrency)\label{def:(Uniformly-distributed-fee-uncon}.
A cryptocurrency whose distribution of mining and/or transaction fees
is uniformly distributed among all the transaction processing nodes,
and doesn't require significant investments from the participating
miners.
\end{defn}

\begin{example}
The proposed cryptocurrency using Zero-Knowledge Proof-of-Identity
is a Uniformly-Distributed Capital-Efficient cryptocurrency.
\end{example}

\begin{defn}
(Game of Rational Mining of Cryptocurrencies)\label{def:(Game-of-Rational}.
A rational miner ranks the cryptocurrencies according to their expected
mining difficulty, and chooses to mine those with lowest expected
difficulty.
\end{defn}

\begin{example}
Awesome Miner\cite{awesomeMiner}, MinerGate\cite{minerGate}, MultiMiner\cite{multiMiner},
MultiPoolMiner\cite{multipoolMiner}, Smart-Miner\cite{smartMiner,smartMinerPaper}
and NiceHash Miner\cite{niceHash} are practical implementations of
the Game of Rational Mining of Cryptocurrencies (although also considering
their prices in addition to their difficulties). Specific calculators
for mining profitability\cite{whatToMine,2CryptoCalc,CoinWarz,CryptoCompare,CryptoZone,Crypto-CoinZ}
could also be used for the similar purposes. Additionally, other papers\cite{modelMinerHashRate}
provide models regarding optimal hash rate allocation.
\end{example}

Let $u_{i}$ be the payoff or utility function for each miner, expressing
his payoff in terms of the decisions or strategies $s_{i}$ of all
the miners,
\[
u_{i}\left(s_{1},s_{2},\ldots,s_{n}\right)=u_{i}\left(s_{i},s_{-i}\right)
\]
whese $s_{-i}$ are set of the strategies of the rest of miners,
\[
s_{-i}=\left(s_{1},s_{2},\ldots,s_{i-1},s_{i+1},\ldots,s_{n}\right)
\]

\begin{defn}
A strategy $s_{1}$ \textit{strictly dominates a strategy $s_{2}$
for miner $i$} if and only if, for any $s_{-i}$ that miner $i$'s
opponents might use,
\[
u_{i}\left(s_{1},s_{-i}\right)>u_{i}\left(s_{2},s_{-i}\right)
\]
\end{defn}

That is, no matter what the other miners do, playing $s_{1}$ is strictly
better than playing $s_{2}$ for miner $i$. Conversely, we say that
the strategy \textit{$s_{2}$ is strictly dominated by $s_{1}$}:
a rational miner $i$ would never play a strictly dominated strategy.
\begin{defn}
A strategy $s_{i}^{*}$ is a \textit{strictly dominant strategy for
miner $i$} if and only if, for any profile of opponent strategies
$s_{-i}$ and any other strategy $s_{i}^{'}$ that miner $i$ could
choose,
\[
u_{i}\left(s_{i}^{*},s_{-i}\right)>u_{i}\left(s_{i}^{'},s_{-i}\right)
\]
\end{defn}

We now demonstrate that mining \textit{UDCE crypto-cryptocurrencies}
\ref{def:(Uniformly-distributed-fee-uncon} is a strictly dominant
strategy with regard to \textit{PLFC cryptocurrencies \ref{def:(Power-law-fee-concentrated-cryp}}
in the \textit{Game of Rational Mining of Cryptocurrencies} \ref{def:(Game-of-Rational}
by showing that every miner's expected profitability is higher in
UDCE cryptocurrencies.
\begin{thm}
\label{thm:Uniformly-distributed-crypto-cur}UDCE cryptocurrencies
yield a strictly higher miner's expected profitability compared to
PLFC cryptocurrencies in the Game of Rational Mining of Cryptocurrencies.
\end{thm}

\begin{proof}
Let $N$ be the number of miners and $R$ the average daily minted
reward per day: UDCE cryptocurrencies award an average of $\nicefrac{N}{R}$
units of cryptocurrency to every participant miner. For every miner
on the long tail of the power distribution, the amount earned with
UDCE cryptocurrencies is obviously higher than with PLFC cryptocurrencies.
For the few miners that dominate PLFC cryptocurrencies, their profitability
is lower because they have to account for the energy\cite{natureEnergyCarbonCosts,bitcoinCarbonFootprint}
and equipment costs in the case of PoW cryptocurrencies or the opportunity
cost of staking capital in volatile PoS cryptocurrencies\cite{stakedPoS},
meanwhile in UDCE cryptocurrencies their cost of mining is so negligible
compared to PLFC cryptocurrencies that the balance of profitability
is always tipped in their favour.
\end{proof}
\begin{defn}
\label{def:The-process-to}The process to solve games called \textit{Iterated
Deletion of Strictly Dominated Strategies (IDSDS)} is defined by the
next steps:

\begin{enumerate}
\item For each player, eliminate all strictly dominated strategies.
\item If any strategy was deleted during Step 1, repeat Step 1. Otherwise,
stop.
\end{enumerate}
\end{defn}

If the process eliminates all but one unique strategy profile $s^{*}$,
we say it is the \textit{outcome of iterated deletion of strictly
dominated strategies} or a \textit{dominant strategy equilibrium}.
\begin{defn}
\label{def:pure-nash}A strategy profile $s^{*}=\left(s_{1}^{*},s_{2}^{*},\ldots,s_{n}^{*}\right)$
is a \textit{Pure-Strategy Nash Equilibrium} (PSNE) if, for every
player $i$ and any other strategy $s_{i}^{'}$ that player $i$ could
choose,
\[
u_{i}\left(s_{i}^{*},s_{-i}^{*}\right)\geq u_{i}\left(s_{i}^{'},s_{-i}^{*}\right)
\]
\end{defn}

~
\begin{defn}
\label{def:strict-nash}A strategy profile $s^{*}=\left(s_{1}^{*},s_{2}^{*},\ldots,s_{n}^{*}\right)$
is a \textit{Strict Nash Equilibrium} (SNE) if, for every player $i$
and any other strategy $s_{i}^{'}$ that player $i$ could choose,
\[
u_{i}\left(s_{i}^{*},s_{-i}^{*}\right)>u_{i}\left(s_{i}^{'},s_{-i}^{*}\right)
\]
\end{defn}

Additionally, if a game is solvable by Iterated Deletion of Strictly
Dominated Strategies, the outcome is a Nash equilibrium.
\begin{thm}
\label{thm:UDCE-cryptocurrency-dominates}A UDCE cryptocurrency dominating
PLFC cryptocurrencies is a Nash equilibrium.
\end{thm}

\begin{proof}
Mining a UDCE cryptocurrency is a strictly dominant strategy with
regard to other miners mining PLFC cryptocurrencies because PLFC cryptocurrencies
are strictly dominated by UCDE cryptocurrencies by Theorem \ref{thm:Uniformly-distributed-crypto-cur},
thus a rational miner will always prefer to miner the latter.

Thus, by the application of Iterated Deletion of Strictly Dominated
Strategies (IDSDS) to the Game of Rational Mining of Cryptocurrencies
\ref{def:(Game-of-Rational}, each miner will eliminate mining PLFC
cryptocurrencies in favor of mining an UDCE cryptocurrency, leaving
this as the unique outcome: therefore, mining a UDCE cryptocurrency
is a \textit{dominant strategy equilibrium} by Definition \ref{def:The-process-to}
and a \textit{Nash equilibrium} by Definition \ref{def:pure-nash}
or by Definition \ref{def:strict-nash}.
\end{proof}
\begin{claim}
\label{claim:(Uniqueness-of-Technical-Solution)}(Uniqueness of Technical
Solution). The proposed technical solution using Zero-Knowledge Proof
of Identity from trusted public certificates (i.e., national identity
cards and/or ePassports) is the only practical and unique solution
for a UCDE cryptocurrency.
\end{claim}

\begin{proof}
As demonstrated in the paper describing ``The Sybil Attack''\cite{the-sybil-attack},
Sybil attacks are always possible unless a trusted identification
agency certifies identities.

As National Identity Cards and ePassports are the only globally available
source of trusted cryptographic identity (3.5B for National Identity
Cards and 1B for ePassports), the only way to bootstrap a UCDE cryptocurrency
is by using the proposed Zero-Knowledge Proof-of-Identity from trusted
public certificates (National Identity Cards and/or ePassports).
\end{proof}

\subsubsection{Strictly Dominant Cryptocurrencies and Evolutionary Stable Strategies\label{subsec:Evolutionary-Stable-Strategies}}

Another interesting viewpoint to consider in the analysis of the cryptocurrency
market is the one offered by behavioural ecology and its Evolutionary
Stable Strategies \ref{def:ESS}: each cryptocurrency can be considered
a unique individual in a population, genetically programmed to play
a pre-defined strategy. New cryptocurrencies are introduced into the
population as individuals with different mutations that define their
technical features (e.g., forking the code of a cryptocurrency to
change the hashing algorithm, or a zk-PoI cryptocurrency). An Evolutionary
Stable Strategy \ref{def:ESS} is a strategy that cannot be invaded
by any alternative strategy, that is, it can resist to the invasion
of a mutant and it's impenetrable to them: once it's introduced and
becomes dominant in a population, natural selection is sufficient
to prevent invasions from new mutant strategies.
\begin{defn}
\label{def:ESS}The pure strategy $s^{*}$ is an Evolutionary Stable
Strategy\cite{logicAnimalConflict} if there exists $\epsilon_{0}>0$
such that:
\[
\left(1-\epsilon\right)\left(u\left(s^{*},s^{*}\right)\right)+\epsilon\left(u\left(s^{*},s^{'}\right)\right)>\left(1-\epsilon\right)\left(u\left(s^{'},s^{*}\right)\right)+\epsilon\left(u\left(s^{'},s^{'}\right)\right)
\]
for all possible deviations $s^{'}$and for all mutation sizes $\epsilon<\epsilon_{0}$.
There are two conditions for a strategy $s^{*}$ to be an Evolutionary
Stable Strategy: for all $s^{*}\neq s^{'}$ either 
\end{defn}

\begin{enumerate}
\item $u\left(s^{*},s^{*}\right)>u\left(s^{'},s^{*}\right)$, that is, it's
a Strict Nash Equilibrium \ref{def:strict-nash}, \textbf{or},
\item if $u\left(s^{*},s^{*}\right)=u\left(s^{'},s^{*}\right)$ then $u\left(s^{*},s^{'}\right)>u\left(s^{'},s^{'}\right)$
\end{enumerate}
\begin{thm}
Mining a UDCE cryptocurrency is an Evolutionary Stable Strategy.
\end{thm}

\begin{proof}
Since mining a UCDE cryptocurrency is a strictly-dominant strategy
and a Strict Nash Equilibrium \ref{thm:UDCE-cryptocurrency-dominates},
then it is an Evolutionary Stable Strategy because it fulfills its
first condition \ref{def:ESS}.

Additionally, mining a UCDE cryptocurrency based on the global network
of National Identity Cards and ePassports is an Evolutionary Stable
Strategy over national variants/mutants due to Claim \ref{claim:(Uniqueness-of-Technical-Solution)}.
\end{proof}
Thus, the Game of Rational Mining of Cryptocurrencies \ref{def:(Game-of-Rational}
is a ``survival of the fittest'' ecology, where the cheapest cryptocurrency
to mine offering the higher profits rises above the others.

\subsubsection{Obviating the Price of Crypto-Anarchy\label{subsec:Price-of-Crypto-Anarchy}}

The most cost efficient Sybil-resistant mechanism is the one provided
by a trusted PKI infrastructure\cite{the-sybil-attack} and a centralised
social planner would prefer the use of National Identity Cards and/or
ePassports in order to minimise costs: instead, permissionless blockchains
are paying very high costs by using PoW/PoS as Sybil-resistant mechanisms.
In this paper, Zero-Knowledge Proof-of-Identity is introduced as a
compromise solution between both approaches, thus obtaining a very
efficient Sybil-resistant mechanism with the best of both worlds.

In order to measure how the efficiency of a Sybil-resistant mechanism
degrades due to the selfish behaviour of its agents (i.e., a fixed
amount of block reward to be distributed among a growing and unbounded
number of miners paying high energy costs, as in Bitcoin), we compare
the ratio between the worst Nash equilibrium and the optimal centralised
solution, a concept known as Price of Anarchy in game theory because
it bounds and quantifies the costs of the selfish behavior of the
agents.
\begin{defn}
\textit{The Price of Anarchy}\cite{DBLP:conf/stacs/KoutsoupiasP99}.
Consider a game $G=\left(N,S,u\right)$ defined as a set of players
$N$, strategy sets $S_{i}$ for each player and utilities $u_{i}:S\rightarrow\mathbb{R}$
(where $S=S_{1}x\ldots xS_{n}$ are also called the set of outcomes).
Define a measure of efficiency of each outcome that we want to minimise,
$Cost:S\rightarrow\mathbb{R}$, and let $Equil\subseteq S$ be the
set of strategies in Nash equilibria. The \textit{Price of Anarchy}
is given by the following ratio:
\[
\mbox{Price of Anarchy}=\frac{\max_{s\in Equil}Cost\left(s\right)}{\min_{s\in S}Cost\left(s\right)}
\]
\end{defn}

The competition game between several blockchains and their cryptocurrencies
can be reformulated\cite{altman:hal-01906954} as a congestion game\cite{congestionGames,potentialGames}
(hereby included for completeness), more amenable to the formulations
commonly used for analyzing the Price of Anarchy (the necessity for
the following definitions is already intuited in the Discussion of
\cite{sokGameTheoryCryptocurrencies}): as the number of miners increases,
it also exponentially decreases the chance that a given miner wins
the block reward by being the first to solve the mining puzzle (i.e.,
the system becomes increasingly congested); it has been proved that
free entry is solely responsible for determining the resource usage\cite{bitcoinMarketStructure},
and that the difficulty is not an instrument that can regulate it.

\paragraph{Miners, mining servers and crypto-currencies}

Denote by $\mathcal{N}\coloneqq\left\{ 1,2,\ldots,N\right\} $ the
finite set of miners that alter the utilities of other miners if any
of them change strategies and let $\mathcal{K}\coloneqq\left\{ 1,2,\ldots,K\right\} $
be the set of cryptocurrencies, each associated to exactly one puzzle
that miners are trying to solve. Let $\mathcal{M}\coloneqq\left\{ 1,2,\ldots,M\right\} $denote
the set of Edge computing Service Providers (ESPs), or mining servers
used to offload the costly computational processing.

\paragraph{Strategies}

Let $\mathcal{S}_{i}\subset\mathcal{K}x\mathcal{M}$ denote the set
of ordered pairs (cryptocurrency, ESP) corresponding to ESPs that
miner $i$ can rely on to solve the puzzles of a given cryptocurrency.
A strategy for miner $i$ is denoted by $s_{i}\in\mathcal{S}_{i}$
corresponding to the cryptocurrency (puzzle) which a miner intends
to solve using a given infrastructure. A strategy vector $s\coloneqq\left(s_{i}\right)_{i\in\mathcal{N}}$
produces a load vector $l\coloneqq\left(l_{k,m}\right)_{k,m}$, where
$l_{k,m}$ denotes the number of users mining blockchain $k$ at ESP
$m$.

\paragraph{Rewards, costs, and utilities}

Let $\eta_{k}$ be the load of miners across all ESPs towards cryptocurrency
$k.$ Then,
\[
\eta_{k}\coloneqq\sum_{m^{'}\in\mathcal{M}}l_{k,m'}\mu_{k,m'}
\]
For a given load vector $l$, the time to solve the puzzle of the
$k^{th}$ cryptocurrency is exponentially distributed with expectation
$\nicefrac{1}{\eta_{k}}$. Let $q_{k}$ be the probability that puzzle
$k$ is solved by time $T$,
\[
q_{k}=1-\mbox{exp}\left(-T\eta_{k}\right)
\]
The probability that a given miner using ESP $m$ is the first to
solve puzzle $k$ is
\[
p_{k,m}=1_{l_{k,m>0}}q_{k}\mu_{k,m}/\eta_{k}
\]
where $1_{c}$ equals 1 if condition $c$ holds and 0 otherwise. For
simplification, subscript $m$ can be dropped and we consider a single
ESP. Then, the probability that a miner is the first to solve the
puzzle is
\[
p_{k}\left(l_{k}\right)=\left(1-\mbox{exp}\left(-T\mu_{k}l_{k}\right)\right)/l_{k}
\]

Let $U_{k,m}\left(l\right)$ denote the utility to a miner who tries
to find the solution of the current puzzle associated to cryptocurrency
$k$ using ESP $m$ and $\gamma_{k,m}$ denote the cost of mining
blockchain $k$ at ESP $m$:
\[
U_{k,m}\left(l\right)=\begin{cases}
p_{k,m}-\gamma_{k,m} & \mbox{if }p_{k,m}>\gamma_{k,m},\\
0 & \mbox{otherwise}
\end{cases}
\]
and the utility of a tagged miner to mine a cryptocurrency $k$ when
there are $l_{k}$ miners associated with the same cryptocurrency
is
\[
U_{k}\left(l_{k}\right)=p_{k}-\gamma_{k},\mbox{ if }p_{k}-\gamma_{k}\geq0
\]

\begin{thm}
\label{thm:altmanTheo}\cite{altman:hal-01906954}If for all $i$
and $j$, $S_{i}=S_{j}$ and $s_{i}$ does not depend on $i$, then
the Nash equilibrium is given by the solution of the following optimization
problem,
\begin{eqnarray*}
\mbox{argmin}_{s}\Phi\left(s\right) & \coloneqq & \sum_{k\in\mathcal{K}}\sum_{l=1}^{l_{k}}p_{k}\left(l\right)-\gamma_{k}\\
\mbox{subject to:} &  & \sum_{k\in\mathcal{K}}l_{k}\leq N,l_{k}\geq0
\end{eqnarray*}
\end{thm}

\begin{defn}
Let $NashCongestedEquil\subseteq S$ be the set of strategies given
as solution of the optimization problem of Theorem \ref{thm:altmanTheo},
then the \textit{Price of Crypto-Anarchy} is given by the following
ratio:
\[
\mbox{Price of Crypto-Anarchy=\ensuremath{\frac{\mbox{max}_{\mbox{\ensuremath{s}}\in\mbox{\ensuremath{NashCongestedEquil}}}\mbox{\ensuremath{Cost\left(s\right)}}}{\mbox{\ensuremath{Cost\left(\mbox{zk-PoI}\right)}}}}}
\]
\end{defn}

In practice, the real-world costs of the Zero-Knowledge Proof of Identity
can be considered almost zero because it's subsidised by governments
and thus exogenous to any blockchain system. Quite the opposite, the
energy costs of PoW cryptocurrencies are notoriously high\cite{natureEnergyCarbonCosts,bitcoinCarbonFootprint}:
it is estimated that mining Bitcoin, Ethereum, Litecoin and Monero
consumed an average of 17, 7, 7 and 14 MJ to generate one US\$, respectively;
and that Bitcoin causes around 22 megatons in CO2 emissions annually\cite{bitcoinCarbonFootprint}.

The trivial extension to Proof-of-Stake is left as an exercise to
the reader, although it's not as affordable as it may be seen: as
of March 2019, an average of 40\% of the cryptocurrency supply is
staked at a total of \$4Bn between all PoS blockchains\cite{stakedPoS}.
Actually, Proof-of-Stake is not strictly better than Proof-of-Work
as the distribution of the market shares between both technologies
has been shown to be indistinguishable (Appendix 3, \cite{2017arXiv170505334E}).

\subsubsection{Pareto Dominance on Currency Circulation\label{subsec:Pareto-Dominance-on}}

For completeness, a stylised version of a model of competing currencies\cite{RePEc:nbr:nberwo:22157}
is introduced here to prove that UDCE cryptocurrencies also dominate
PLFC cryptocurrencies on their circulation (i.e., trading, speculating)
due to their stronger network effects, and not only mining as previously
proved. The key observation here is that by definition \ref{def:(Power-law-fee-concentrated-cryp},
the returns of mining PLFC cryptocurrencies is concentrated on a very
limited number of miners and the newly minted cryptocurrency could
be held for long periods of times: otherwise, if they didn't expect
that the held cryptocurrencies would appreciate in time, they would
be mining another set of cryptocurrencies with better expectations.
In direct contrast, the distribution of mining and/or transaction
fees of UDCE cryptocurrencies is uniformly distributed by definition
\ref{def:(Uniformly-distributed-fee-uncon}: therefore, the returns
of the holding strategy after minting them would be lower and their
subsequent circulation much less restricted.

Suppose an economy divided into periods, each period divided into
two subperiods: in the first subperiod, a perishable good demanded
by everyone is produced and consumed in a Centralised Market; in the
second subperiod, buyers who only consume are randomly matched with
sellers who only produce with probability $\sigma\in\left(0,1\right)$
in a Decentralised Market. Let $\beta\in\left(0,1\right)$ denote
the discount factor, $\phi_{t}^{i}\in\mathbb{R}_{+}$ denote the value
of a unit of currency $i\in\left\{ 1,\ldots,N\right\} $ in terms
of the CM food and $\phi_{t}=\left(\phi_{t}^{1},\ldots,\phi_{t}^{N}\right)\in\mathbb{R}_{+}^{N}$
denote the vector of real prices.
\begin{defn}
(Buyers). In a $\left[0,1\right]$-continuum of buyers, $x_{t}^{b}\in\mathbb{R}$
denotes the buyer's net consumption of the CM good and $q_{t}\in\mathbb{R}_{+}$denotes
the consumption of the DM good. The utility function of the buyer's
preferences is given by
\[
U^{b}\left(x_{t}^{b},q_{t}\right)=x_{t}^{b}+u\left(q_{t}\right)
\]
with $u:\mathbb{R}_{+}\rightarrow\mathbb{R}$ continuously differentiable,
increasing and strictly concave with $u'\left(0\right)=\infty$ and
$u\left(0\right)=0$.

Let $W^{b}\left(M_{t-1}^{b},t\right)$ denote the value function for
a buyer who starts period $t$ holding a portfolio $M_{t-1}^{b}\in\mathbb{R}_{+}^{N}$
of cryptocurrencies in the CM and let $V^{b}\left(M_{t}^{b},t\right)$
denote the value function in the DM: the dynamic programming equation
is
\[
W^{b}\left(M_{t-1}^{b},t\right)=\underset{\left(x_{t}^{b},M_{t}^{b}\right)\in\mathbb{R}\times\mathbb{R}_{+}^{N}}{\mbox{max}}\left[x_{t}^{b}+V^{b}\left(M_{t}^{b},t\right)\right]
\]
subject to the budget constraint
\[
\phi_{t}\cdot M_{t}^{b}+x_{t}^{b}=\phi_{t}\cdot M_{t-1}^{b}.
\]
The value for a buyer holding a portfolio $M_{t}^{b}$ in the DM is
\begin{eqnarray*}
V^{b}\left(M_{t}^{b},t\right) & = & \sigma\left[u\left(q\left(M_{t}^{b},t\right)\right)+\beta W^{b}\left(M_{t}^{b}-d\left(M_{t}^{b},t\right),t+1\right)\right]\\
 &  & +\left(1-\sigma\right)\beta W^{b}\left(M_{t}^{b},t+1\right)
\end{eqnarray*}
and let $q^{*}\in\mathbb{R}$ denote the quantity satisfying $u'\left(q^{*}\right)=w'\left(q^{*}\right)$
so that $q^{*}$ gives the surplus-maximizing quantity that determines
the efficient level of production in the DM. The solution to the bargaining
problem is given by 
\[
q\left(M_{t}^{b},t\right)=\begin{cases}
m^{-1}\left(\beta\times\phi_{t+1}\cdot M_{t}^{b}\right) & \mbox{if}\,\phi_{t+1}\cdot M_{t}^{b}<\beta^{-1}\left[\theta w\left(q^{*}\right)+\left(1-\theta\right)u\left(q^{*}\right)\right]\\
q^{*} & \mbox{if}\,\phi_{t+1}\cdot M_{t}^{b}\geq\beta^{-1}\left[\theta w\left(q^{*}\right)+\left(1-\theta\right)u\left(q^{*}\right)\right]
\end{cases}
\]
and
\[
\phi_{t+1}\cdot d\left(M_{t}^{b},t\right)=\begin{cases}
\phi_{t+1}\cdot M_{t}^{b} & \mbox{if}\,\phi_{t+1}\cdot M_{t}^{b}<\beta^{-1}\left[\theta w\left(q^{*}\right)\right.\\
 & \left.+\left(1-\theta\right)u\left(q^{*}\right)\right]\\
\beta^{-1}\left[\theta w\left(q^{*}\right)+\left(1-\theta\right)u\left(q^{*}\right)\right] & \mbox{if}\,\phi_{t+1}\cdot M_{t}^{b}\geq\beta^{-1}\left[\theta w\left(q^{*}\right)\right.\\
 & \left.+\left(1-\theta\right)u\left(q^{*}\right)\right]
\end{cases}
\]
with the function $m:\mathbb{\mathbb{R}}_{+}\rightarrow\mathbb{R}$
defined as
\[
m\left(q\right)\equiv\frac{\left(1-\theta\right)u\left(q\right)w'\left(q\right)+\theta w\left(q\right)u'\left(q\right)}{\theta u'\left(q\right)+\left(1-\theta\right)w'\left(q\right)}.
\]
The optimal portfolio problem can be defined as
\[
\underset{M_{t}^{b}\in\mathbb{R}_{+}^{N}}{\mbox{max}}\left\{ -\phi_{t}\cdot M_{t}^{b}+\sigma\left[u\left(q\left(M_{t}^{b},t\right)\right)-\beta\times\phi_{t+1}\cdot d\left(M_{t}^{b},t\right)\right]+\beta\times\phi_{t+1}\cdot M_{t}^{b}\right\} 
\]
thus the optimal choice satisfies
\begin{equation}
\phi_{t}^{i}=\beta\phi_{t+1}^{i}L_{\theta}\left(\phi_{t+1}\cdot M_{t}^{b}\right)\label{eq:1}
\end{equation}
for every type $i\in\left\{ 1,\ldots,N\right\} $ together with the
transversality condition
\begin{equation}
\underset{t\rightarrow\infty}{\mbox{lim}}\beta^{t}\times\phi_{t}\cdot M_{t}^{b}=0\label{eq:2}
\end{equation}
where $L_{\theta}:\mathbb{R}_{+}\rightarrow\mathbb{R}_{+}$ is given
by
\[
L_{\theta}\left(A\right)=\begin{cases}
\sigma\frac{u'\left(m^{-1}\left(\beta A\right)\right)}{w'\left(m^{-1}\left(\beta A\right)\right)}+1-\sigma & \mbox{if}\,A<\beta^{-1}\left[\theta w\left(q^{*}\right)+\left(1-\theta\right)u\left(q^{*}\right)\right]\\
1 & \mbox{if}\,A\geq\beta^{-1}\left[\theta w\left(q^{*}\right)+\left(1-\theta\right)u\left(q^{*}\right)\right]
\end{cases}
\]
\end{defn}

~
\begin{defn}
(Sellers). In a $\left[0,1\right]$-continuum of sellers, $x_{t}^{s}\in\mathbb{R}$
denotes the seller's net consumption of the CM good and $n_{t}\in\mathbb{R}_{+}$denotes
the seller's effort level to produce the DM good. The utility function
of the seller's preferences is given by
\[
U^{s}\left(x_{t}^{s},n_{t}\right)=x_{t}^{s}-w\left(n_{t}\right)
\]
with $w:\mathbb{R}_{+}\rightarrow\mathbb{R}_{+}$ continuously differentiable,
increasing and weakly convex with $w\left(0\right)=0$.

Let $W^{s}\left(M_{t-1}^{s},t\right)$ denote the value function for
a seller who enters period $t$ holding a portfolio $M_{t-1}^{s}\in\mathbb{R}_{+}^{N}$
of cryptocurrencies in the CM and let $V^{s}\left(M_{t}^{s},t\right)$
denote the value function in the DM: the dynamic programming equation
is
\[
W^{s}\left(M_{t-1}^{s}\right)=\underset{\left(x_{t}^{s},M_{t}^{s}\right)\in\mathbb{R}\times\mathbb{R}_{+}^{N}}{\mbox{max}}\left[x_{t}^{s}+V^{s}\left(M_{t}^{s},t\right)\right]
\]
subject to the budget constraint
\[
\phi_{t}\cdot M_{t}^{s}+x_{t}^{s}=\phi_{t}\cdot M_{t-1}^{s}.
\]
The value for a seller holding a portfolio $M_{t}^{s}$ in the DM
is
\begin{eqnarray*}
V^{s}\left(M_{t}^{s},t\right) & = & \sigma\left[-w\left(q\left(M_{t}^{b},t\right)\right)+\beta W^{s}\left(M_{t}^{s}+d\left(M_{t}^{b},t\right),t+1\right)\right]\\
 &  & +\left(1-\sigma\right)\beta W^{s}\left(M_{t}^{s},t+1\right)
\end{eqnarray*}
\end{defn}

~
\begin{defn}
(Miners). In a $\left[0,1\right]$-continuum of miners of each type-$i\in\left\{ 1,\ldots,N\right\} $
token, $x_{t}^{i}\in\mathbb{R}_{+}$ denotes the miner's consumption
of the CM good and $\triangle_{t}^{i}\in\mathbb{R}_{+}$denotes the
production of the type-$i$ token. The utility function of the miner's
preferences is given by
\[
U^{e}\left(x_{t}^{i},\Delta_{t}\right)=x_{t}^{i}-c\left(\Delta_{t}^{i}\right)
\]
with the cost function $c:\mathbb{R}_{+}\rightarrow\mathbb{R}_{+}$
continuously differentiable, strictly increasing and weakly convex
with $c\left(0\right)=0$.

Let $M_{t}^{i}\in\mathbb{R}_{+}$ denote the per-capita supply of
cryptocurrency $i$ in period $t$ and $\Delta_{t}^{i}\in\mathbb{R}$
denote the miner $i$'s net circulation of newly minted tokens in
period $t$. To describe the miner's problem to determine the money
supply in the economy, we start assuming that all miners solve the
same decision problem, thus the law of motion of type-$i$ tokens
at all date $t\geq0$ is given by

\begin{equation}
M_{t}^{i}=\Delta_{t}^{i}+M_{t-1}^{i}\label{eq:4}
\end{equation}
where $M_{-1}^{i}\in\mathbb{R}_{+}$ denotes the initial stock. The
budget constraint is
\[
x_{j}^{i}=\phi_{t}^{i}\Delta_{t}^{i},
\]
and given that the miner takes prices $\left\{ \phi_{t}\right\} _{t=0}^{\infty}$
as given, the profit maximization of the cryptocurrency emission problem
is solved by
\begin{equation}
\Delta_{t}^{*,i}\in\underset{\Delta\in\mathbb{R}_{+}}{\mbox{arg max}}\left[\phi_{t}^{i}\Delta-c\left(\Delta\right)\right]\label{eq:3}
\end{equation}
\end{defn}

~
\begin{defn}
(Equilibrium). A perfect-foresight monetary equilibrium is an array
$\left\{ M_{t},M_{t}^{b},\Delta_{t}^{*},\phi_{t}\right\} _{t=0}^{\infty}$
satisfying \ref{eq:1}, \ref{eq:2}, \ref{eq:3} and \ref{eq:4} for
each $i\in\left\{ 1,\ldots,N\right\} $ at all dates $t\geq0$ and
satisfying the following market-clearing condition
\begin{equation}
M_{t}=M_{t}^{b}+M_{t}^{s}
\end{equation}
\end{defn}

Suppose that each miner $j$ starts with $M^{i}>0$ units of currency
$i\in\left\{ 1,\ldots,N\right\} $: let $\delta$ denote the fraction
$1-\delta$ of randomly selected miners in each location $j$ at each
date $t\geq0$ who doesn't offer their tokens to sellers because they
are holding them in expectation of their appreciation (i.e., PLFC
cryptocurrencies), so these tokens don't circulate to other $j$ positions
whenever sellers are relocated.

Conversely, an equilibrium with the property that miners don't restrict
the circulation of recently mined tokens (i.e., UDCE cryptocurrencies)
is as follows: the optimal portfolio choice implies the first-order
condition
\[
\frac{u'\left(q\left(M_{t},t\right)\right)}{w'\left(q\left(M_{t},t\right)\right)}=\frac{1}{\beta\gamma_{t+1}^{i}}
\]
for each currency $i$, where $\gamma_{t+1}\in\mathbb{R}_{+}$ represents
the common return across all valued currencies between dates $t$
and $t+1$. The demand for real balances in each location is given
by
\[
z\left(\gamma_{t+1};1\right)\equiv\frac{1}{\gamma_{t+1}}L_{1}^{-1}\left(\frac{1}{\beta\gamma_{t+1}}\right)
\]
because 
\[
\beta\gamma_{t+1}\sum_{i=1}^{N}b_{t}^{i}<\theta w\left(q^{*}\right)+\left(1-\theta\right)u\left(q^{*}\right)
\]
and with 
\[
L_{\delta}\left(A\right)=\begin{cases}
\delta\frac{u'\left(m^{-1}\left(\beta A\right)\right)}{w'\left(m^{-1}\left(\beta A\right)\right)}+1-\delta & \mbox{if}\,A<\beta^{-1}w\left(q^{*}\right)\\
1 & \mbox{if}\,A\geq\beta^{-1}w\left(q^{*}\right)
\end{cases}
\]
Because the market-clearing condition implies
\[
\sum_{i=1}^{N}\phi_{t}^{i}M^{i}=z\left(\gamma_{t+1};1\right)
\]
the equilibrium sequence $\left\{ \gamma_{t}\right\} _{t=0}^{\infty}$
satisfies the law of motion
\[
z\left(\gamma_{t+1};1\right)=\gamma_{t}z\left(\gamma_{t};1\right)
\]
because 
\[
M_{t}^{i}=M_{t-1}^{i}=\Delta^{i}
\]
for each $i$ and provided that $\gamma_{t}\leq t$, and the boundary
condition
\[
\beta\gamma_{t}z\left(\gamma_{t};1\right)\leq w\left(q^{*}\right)
\]
Suppose $u\left(q\right)=\left(1-\eta\right)^{-1}q^{1-\eta}$, with
$0<\eta<1$, and $w\left(q\right)=\left(1+\alpha\right)^{-1}q^{1+\alpha}$
with $\alpha\geq0$. The dynamic system describing the equilibrium
evolution of $\gamma_{t}$ is 
\begin{equation}
\gamma_{t+1}^{\frac{1+\alpha}{\eta+\alpha}-1}=\gamma_{t}^{\frac{1+\alpha}{\eta+\alpha}}\label{eq:28}
\end{equation}

\begin{thm}
The allocation associated with the circulation of UDCE cryptocurrencies
Pareto dominates the allocation with the associated the circulation
of PLFC cryptocurrencies, on a stationary equilibrium with the property
that the quantity traded in the Decentralised Market is given by $\hat{q}\left(1\right)\in\left(\hat{q}\left(\delta\right),q^{*}\right)$
satisfying
\begin{equation}
\frac{u'\left(\hat{q}\left(1\right)\right)}{w'\left(\hat{q}\left(1\right)\right)}=\beta^{-1}\label{eq:29}
\end{equation}
\end{thm}

\begin{proof}
The sequence $\gamma_{t}=1$ for all $t\geq0$ satifies \ref{eq:28}.
Then, the solution to the optimal portfolio problem implies that the
DM output must satisfy \ref{eq:29}. The quantity $\hat{q}$ satisfies
\[
\frac{u'\left(\hat{q}\left(1\right)\right)}{w'\left(\hat{q}\left(1\right)\right)}=\delta\frac{u'\left(\hat{q}\left(\delta\right)\right)}{w'\left(\hat{q}\left(\delta\right)\right)}+1-\delta
\]
Because $\delta\in\left(0,1\right)$, we have $\hat{q}\left(1\right)>\hat{q}\left(\delta\right)$,
that is, the allocation associated with the circulation of UCDE cryptocurrencies
-$\hat{q}\left(1\right)$- Pareto dominates the allocation associated
with the circulation of PLFC cryptocurrencies -$\hat{q}\left(\delta\right)$-.
\end{proof}

\subsubsection{On Network Effects\label{subsec:On-Network-Effects}}

At the time of the release of this paper, cryptocurrencies have failed
to provide an alternative to traditional payment networks due to a
combination of high transaction fees, high finalization time and high
volatility. The failure to find the favor of merchants is also their
biggest weakness: they aren't part of two-sided networks, and thus
easily replaceable by any newer cryptocurrency better able to create
them. Actually, the first-mover advantage of the most valued cryptocurrencies
is lower than expected if any competing cryptocurrency leverages network
effects from other different sources (e.g., Zero-Knowledge Proof-of-Identity
from trusted PKI certificates).

A simple model is introduced here to analyze the evolution of competing
payment networks: consider the two-sided and incompatible payment
networks of two cryptocurrencies, BTC and zk-PoI, each with their
corresponding groups of merchants $m$ and customers $c$; let $m_{BTC}^{t},m_{zkPOI}^{t}$
denote the number of merchants at time $t$ and $c_{BTC}^{t},c_{zkPOI}^{t}$
the number of customers. A user joins the payment networks at each
time step $t$, with $\lambda$ being the probability of being a customer
and $1-\lambda$ of being a merchant: each merchant prefers to join
BTC or zk-PoI depending on the number of customers in the same network,
thus the probabilities to join one of the networks are given by
\[
\frac{c_{BTC}^{\beta}}{c_{BTC}^{\beta}+c_{zkPOI}^{\beta}},\frac{c_{zkPOI}^{\beta}}{c_{BTC}^{\beta}+c_{zkPOI}^{\beta}}
\]
 and conversely, for customers the probabilities are given by
\[
\frac{m_{BTC}^{\alpha}}{m_{BTC}^{\alpha}+m_{zkPOI}^{\alpha}},\frac{m_{zkPOI}^{\alpha}}{m_{BTC}^{\alpha}+m_{zkPOI}^{\alpha}}.
\]
Note that some categories of users would prefer to use the expected
number of users and not their current tally: forward-looking merchants
that need to invest on equipment to access the payment network are
within this group, thus they would prefer to use expected numbers,
\[
\frac{E\left(c_{BTC}^{\beta}\right)}{E\left(c_{BTC}^{\beta}\right)+E\left(c_{zkPOI}^{\beta}\right)},\frac{E\left(c_{zkPOI}^{\beta}\right)}{E\left(c_{BTC}^{\beta}\right)+E\left(c_{zkPOI}^{\beta}\right)}.
\]
Each user can only join one payment network, modelling the fact that
single-homing is preferred to multi-homing in the real-world, and
the particular network is determined by the distribution of users
on the other side at each time $t$. The parameters $\alpha,\beta>0$
are elasticities of demand with regard to the numbers of users on
the other side of the payment network, effectively acting as measures
of indirect network effects: $\alpha$ can be empirically estimated
by observing joining customers over a small period of time and then
calculating
\[
\alpha=\frac{\mbox{ln }\left(\nicefrac{m_{BTC}^{\alpha}}{\left(m_{BTC}^{\alpha}+m_{zkPOI}^{\alpha}\right)}\right)-\mbox{ln }\left(1-\left(\nicefrac{m_{BTC}^{\alpha}}{\left(m_{BTC}^{\alpha}+m_{zkPOI}^{\alpha}\right)}\right)\right)}{\mbox{ln }m_{BTC}-\mbox{ln }m_{zkPOI}}
\]
and conversely$\beta$ can be empirically estimated by observing joining
merchants and then calculating
\[
\beta=\frac{\mbox{ln }\left(\nicefrac{c_{BTC}^{\beta}}{\left(c_{BTC}^{\beta}+c_{zkPOI}^{\beta}\right)}\right)-\mbox{ln }\left(1-\left(\nicefrac{c_{BTC}^{\beta}}{\left(c_{BTC}^{\beta}+c_{zkPOI}^{\beta}\right)}\right)\right)}{\mbox{ln }c_{BTC}-\mbox{ln }c_{zkPOI}}.
\]

\begin{thm}
(Dominance of the Zero-Knowledge Proof-of-Identity cryptocurrency).
A new cryptocurrency could achieve dominance over previous cryptocurrencies,
overcoming first-mover advantages, if the expected number of accepting
customers would be much higher and the number of merchants using the
previous cryptocurrencies is low.
\end{thm}

\begin{proof}
Note that the number of steps needed for a new cryptocurrency, $m_{zkPOI}$,
to overtake the previous one, $m_{BTC}$, on the number of merchants,
$m_{zkPOI}>m_{BTC}$, is given by
\[
\left(m_{BTC}+1\right)\cdot\left(1-\lambda\right)^{-1}
\]
It's possible for a new cryptocurrency to overtake a previous one
on the number of merchants whenever
\[
E\left(c_{zkPOI}\right)-E\left(c_{BTC}\right)>\left(m_{BTC}+1\right)\cdot\left(1-\lambda\right)^{-1}
\]
and since $m_{BTC}$ is a low number and $E\left(c_{zkPOI}\right)\gg E\left(c_{BTC}\right)$,
it's conceivable that the previous condition could hold.

Now let's consider the results of strong network effects on the final
market shares of both payment networks by examining the following
differential equations,
\[
\frac{d\left(\nicefrac{m_{BTC}}{m_{zkPOI}}\right)}{dt}=\left(1-\lambda\right)\frac{\left(\nicefrac{c_{BTC}}{c_{zkPOI}}\right)^{\beta}-\left(\nicefrac{m_{BTC}}{m_{zkPOI}}\right)}{\left(1+\left(\nicefrac{c_{BTC}}{c_{zkPOI}}\right)^{\beta}\right)m_{zkPOI}}
\]
and
\[
\frac{d\left(\nicefrac{c_{BTC}}{c_{zkPOI}}\right)}{dt}=\lambda\frac{\left(\nicefrac{m_{BTC}}{m_{zkPOI}}\right)^{\alpha}-\left(\nicefrac{c_{BTC}}{c_{zkPOI}}\right)}{\left(1+\left(\nicefrac{m_{BTC}}{m_{zkPOI}}\right)^{\alpha}\right)c_{zkPOI}}
\]

According to the signs of the previous derivatives, when $\alpha\cdot\beta>1$
and $t\rightarrow\infty$, the payment network with even a slight
advantage over the other will end acquiring all the merchants and
customers, for example
\[
\underset{t\rightarrow\infty}{\lim}m_{zkPOI}=\infty,\underset{t\rightarrow\infty}{\lim}c_{zkPOI}=\infty
\]
\[
\underset{t\rightarrow\alpha}{\lim}m_{BTC}=0,\underset{t\rightarrow\alpha}{\lim}c_{BTC}=0
\]
but with $\alpha\cdot\beta<1$, the number of merchants and customers
will equalize
\[
m_{zkPOI}=m_{BTC},c_{zkPOI}=c_{BTC}
\]
thus highlighting the importance of network effects.
\end{proof}

\subsubsection{Dominance over Cash and other Cryptocurrencies\label{subsec:Dominance-over-Cash}}

The dominance of subsection \ref{subsec:Strictly-Dominant-Cryptocurrenci}
is based on mining and subsection \ref{subsec:Pareto-Dominance-on}
extends said dominance to the circulation of currencies: in this subsection,
the dominance will be based on the lower costs of a payment network
of the cryptocurrency using Zero-Knowledge Proof-of-Identity; therefore,
there exists a unique equilibrium in which this payment system dominates.

A recent paper\cite{coordinationElectronicPaymentInstrument} offers
a model based on a version of Lagos-Wright\cite{lagosWright} to explain
the substitution of cash by debit cards or any other non-deferred
electronic payment system incurring a fixed cost $\Omega\left(z\right)\tau$
per each period $\tau$, the cost $\Omega\left(z\right)$ being financed
by imposing fee $\omega$ on each payment where $\omega$ should satisfy
\[
\Omega\left(z\right)=S\left[\theta\omega+\left(1-\theta\right)\omega\right]
\]
and where $S$ denotes the instantaneous measure of electronic payment
transactions, $z$ is the state of development of the economy, $\theta\in\left[0,1\right]$
is the share of cost allocated to a buyer and $\left(1-\theta\right)$
is the share of cost allocated to a seller. In this model, an electronic
payment system can achieve dominance over cash using the solution
concept of iterative elimination of conditionally dominated strategies
whenever the state of development of the economy $z$ is sufficiently
high, and there exists a unique equilibrium in the model such that
agents choose electronic payment transactions when $z$ is strictly
higher than the limiting cut-off function $Z_{\infty}$ of the sequence
of boundaries $Z_{0,}Z_{1,}\ldots$ of regions where an agent chooses
electronic payment transactions regardless of the choices of other
agents. In other words, it's strictly dominant to choose electronic
payments in an economy having sufficiently advanced information technology
so that $\Omega$ is negligible.
\begin{cor}
Since the cost function $\Omega\left(z\right)$ of a UCDE cryptocurrency
based on Zero-Knowledge Proof-of-Identity is much cheaper than of
PoW/PoS cryptocurrencies and other forms of electronic payment because
its cards are already distributed (i.e., de facto subsidised by governments),
there exists a unique equilibrium in the model \cite{coordinationElectronicPaymentInstrument}
such that the agents choose the UCDE cryptocurrency using zk-PoI and
it dominates the other forms of payment.
\end{cor}

\section{Conclusion}

Although all permissionless blockchains critically depend on Proof-of-Work
or Proof-of-Stake to prevent Sybil attacks, their high resource consumption
corroborates their non-scalability and act as a limiting factor to
the general diffusion of blockchains. This paper proposed an alternative
approach that not only doesn't waste resources, it could also help
in the real-world identity challenges faced by permissionless blockchains:
the derivation of anonymous credentials from widely trusted public
PKI certificates.

Additionally, we study the better incentives offered by the proposed
cryptocurrency based on our anonymous authentication scheme: mining
is proved to be incentive-compatible and a strictly dominant strategy
over previous cryptocurrencies, thus a Nash equilibrium over previous
cryptocurrencies and an Evolutionary Stable Strategy; furthermore,
zk-PoI is proved to be optimal because it implements the social optimum,
unlike PoW/PoS cryptocurrencies that are paying the Price of (Crypto-)Anarchy.
The circulation of the proposed cryptocurrency is proved to Pareto
dominate other cryptocurrencies based on its negligible mining costs
and it could also become dominant thanks to stronger network effects;
finally, the lower costs of its infrastructure imply the existence
of a unique equilibrium where it dominates other forms of payment.

\bibliographystyle{alpha}
\bibliography{bib}

\end{document}